\documentclass[11pt]{article}

\pdfoutput=1

\usepackage{a4wide}
\usepackage[utf8]{inputenc}
\usepackage[T1]{fontenc}     
\usepackage[english]{babel} 
\usepackage{amsmath}
\usepackage{amsfonts}
\usepackage{amssymb}
\usepackage{amsthm}
\usepackage{euscript}
\usepackage[pdftex]{graphicx}          
\usepackage{array}
\usepackage{tikz}
\usepackage{url}
\usetikzlibrary{topaths,calc}
\newcommand{\bookhskip}{\hskip 1em plus 0.5em minus 0.4em\relax}

\newtheorem{theo}{Theorem}[section]
\newtheorem{defi}[theo]{Definition}
\newtheorem{prop}[theo]{Proposition}
\newtheorem{coro}[theo]{Corollary}
\newtheorem{lemma}[theo]{Lemma}

\DeclareMathOperator{\rank}{rank}
\DeclareMathOperator{\Ker}{Ker}
\DeclareMathOperator{\im}{Im}

\newcommand{\C}{\mathbb{C}}
\newcommand{\CC}{{\EuScript C}}
\newcommand{\Z}{\mathbb{Z}}
\newcommand{\N}{\mathbb{N}}
\newcommand{\F}{\mathbb{F}}

\newcommand{\X}{\EuScript X}
\newcommand{\e}{\varepsilon}
\newcommand{\h}{\mathcal{H}}
\newcommand{\hn}{\mathcal{H}^{\otimes n}}

\newcommand{\E}{{\EuScript E}}
\newcommand{\Esp}{{\mathbb E}}
\newcommand{\Prob}{{\mathbb P}}
\newcommand{\Pauli}{\mathcal{P}}
\newcommand{\HH}{\mathbf H}
\newcommand{\hh}{\mathbf h}

\newcommand\ket[1]{|#1\rangle}

\newcommand\proj[1]{|#1\rangle \langle #1|}
\newcommand\gen[1]{\langle #1\rangle}

\title{Upper Bounds on the Rate of Low Density Stabilizer Codes for the Quantum Erasure Channel}           
\author{Nicolas Delfosse, Gilles Zémor\\
Institut de Math\'ematiques de Bordeaux UMR 5251, Universit\'e Bordeaux 1,\\
351, cours de la Lib\'eration, F-33405 Talence Cedex, France \\
Email: \{Nicolas.Delfosse, Gilles.Zemor\}@math.u-bordeaux1.fr}
\date{May 31, 2012}

\begin{document}

\maketitle

\begin{abstract}
%We are interested in the performance of quantum LDPC codes.
Using combinatorial arguments, we determine an upper bound on
achievable rates of stabilizer codes used over the quantum erasure
channel. This allows us to recover the no-cloning bound on the
capacity of the quantum erasure channel, $R \leq 1-2p$, for stabilizer codes:
we also derive an improved upper bound of the form $R \leq 1-2p-D(p)$
with a function $D(p)$ that stays positive for $0<p<1/2$ and for any family of
stabilizer codes whose generators have weights bounded from above by a
constant -- low density stabilizer codes.

We obtain an application
to percolation theory for a family of self-dual tilings of the hyperbolic plane.
We associate a family of low density stabilizer codes with appropriate finite
quotients of these tilings. We then relate the probability of
percolation to the probability of a decoding error for these codes on
the quantum erasure channel. 
The application of our upper bound on achievable rates of 
low density stabilizer codes gives rise to an upper bound on the critical probability for these tilings.
\end{abstract}

\section{Introduction}

Low Density Parity Check (LDPC) codes are classical error-correcting
codes originally introduced by Gallager \cite{Ga63}. They come with
highly efficient local iterative decoding schemes, have been 
extensively studied and have proved very successful on a
number of channels.
%capacity approaching codes. They were introduced by Gallager in 1963 \cite{Ga63}.
% and rediscovered in the nineties by Richardson, Urbanke \cite{} and Mackay \cite{}
%This class of codes is particularly well adapted to transmit information with high rate over a noisy channel because we dispose of an efficient decoding algorithm.
Therefore, in the field of quantum communication and quantum
computation, 
it is natural to look into the quantum analog of classical LDPC codes, 
which arguably are stabilizer codes with generators of bounded weight.
We will call such codes low density stabilizer codes or,
following others, refer to them somewhat loosely as quantum LDPC codes.
These include a number of constructions of locally decodable quantum
codes with a topological connection, starting with Kitaev's celebrated
toric code \cite{Ki97}, and other families among which
surfaces codes \cite{BM073}, \cite{Ze09}, color codes \cite{BM06},
\cite{BM072}, and other variants \cite{TZ09}, \cite{CDZ11}. 
Various generalizations to the quantum setting of classical LDPC codes
have also been proposed, e.g. \cite{Al07}, \cite{Al08}, \cite{COT07}, \cite{HMIH07}, \cite{SKR08}.

In the present work we are interested in the performance of quantum
LDPC codes over the quantum erasure channel. Our motivation is
inspired by the classical setting, in which the systematic study of LDPC codes
for the classical erasure channel has led to a better understanding
of the behaviour of LDPC codes for more complicated channels such as the
 binary symmetric channel and the Gaussian channel. This approach has
 hardly been attempted in the quantum setting and it is not
 unreasonable to hope for similar returns in the long run. The quantum
 erasure channel, besides being simpler than the more universal
 depolarizing channel, is also a realistic channel \cite{GBP97}. 
Its capacity is known, it is $1-2p$ \cite{BDS97}, however almost
nothing is known about the performance of quantum LDPC codes over this
channel.
We would like to gain some understanding as to what are the code characteristics needed
to achieve capacity. 

We shall derive a bound on the achievable rate of stabilizer codes as a function
of an upper bound on the weight of the generators of the stabilizer
group.
Equivalently, we will derive an upper bound on the decoding threshold
on the erasure channel for quantum LDPC codes.
This bound will yield the following result: any family of stabilizer codes
that have stabilizer groups with generators of weight bounded by a
constant, cannot achieve the capacity of the quantum erasure channel.
This phenomenon is somewhat analogous to the classical setting
\cite{Ga63} \cite{BKLM02} where it is known that capacity achieving
LDPC codes must have parity-check matrices with growing row weights.

Our result has an unexpected application to percolation theory. Given
an infinite edge-transitive graph, a random subgraph, called the {\em open}
subgraph is considered where every edge is declared open,
independently of the others, with probability $p$.
The central question in percolation theory is the determination of the
critical probability $p_c$, which is the minimum value of $p$ such
that the open connected component of any given edge $e$ is infinite
with non zero probability. For most graphs, computing the critical
probability exactly is usually quite difficult.

We are interested in the critical probability of graphs that make up
regular tilings of the hyperbolic plane. The connection with quantum
erasure correcting codes is that quotients of the infinite tiling of
the hyperbolic plane yield finite graphs (more precisely
combinatorial surfaces) that define quantum LDPC codes (surface
codes): selecting a random erasure pattern for the finite code is the
same as selecting a random subgraph of the finite graph, and the
non-correctable erasure event is very close to the percolation event
on the infinite graph.
This connection was observed for the toric code \cite{DKLP02}, and
this is why the erasure threshold of the toric code coincides with the
critical probability $p_c = 1/2$ in the square lattice.
We shall derive a bound on the erasure decoding threshold for surface
codes that will lead to an upper bound on the critical probability
$p_c$ for hyperbolic tilings. To the best of our knowledge this is the
sharpest presently known such upper bound.

\bigskip
The paper is organized as follows. 
After introductory and background material presented in Section~\ref{section:background}, we
derive an upper bound on achievable rates of stabilizer codes in
Section~\ref{section:stabilizer}. There are two main results in this
section. The first is Theorem~\ref{theo:capacity_stabilizer} which states
that achievable rates of stabilizer codes satisfy a bound of the form
$R\leq 1-2p-D(p)$ for a non-negative function $D(p)$. This
enables one to recover the capacity bound $R\leq 1-2p$ for the quantum
erasure channel for the particular case of stabilizer codes. It also
enables one to derive improved bounds on achievable rates for
particular classes of stabilizer codes. We derive such an improved
bound in Theorem~\ref{theo:stabilizer} for stabilizer codes with
generators of bounded weight.
In Section~\ref{section:CSS(2,m)} we refine the above bound for a particular
class of LDPC CSS codes that make up a family of surface codes. In
Section~\ref{section:percolation}
we apply the refined bound on achievable rates to percolation on
hyperbolic lattices: the main result is
Theorem~\ref{theo:percolation_bound_LDPC}
which is an upper bound on critical probabilities.
Finally, an appendix regroups some technicalities necessary to
complete formal proofs (Appendix~\ref{appendix:concavity_inequality}). 
%and Figures where bounds are plotted (Appendix~\ref{app:figures}).

\section{Background}\label{section:background}

A stabilizer code of parameters $[[n, k]]$ is a subspace of dimension
$2^k$ of the space $\hn = (\C^2)^{\otimes n}$. It is defined as the
set of fixed points of an abelian group of Pauli operators. Below we
go over notation and definitions. This material is quite well-known
but we have felt the need to highlight the properties that we need, in
particular because we shall make extensive use of linear algebra.
For a more precise description of quantum information and quantum
error-correcting codes, see Nielsen and Chuang \cite{NC00} with a
binary point of view close to the one that we adopt here,
or the article of Calderbank, Rains, Shor and Sloane \cite{CRSS98}, for an $\F_4$ point of view.

\subsection{Pauli groups}

A quantum bit or qubit is a vector of $\h = \C^2$. It is the basic
unit of quantum information.
A sequence of $n$ qubits lives in the space $\hn$. 
The classical Pauli operators form a basis of the space of operators on $\hn$.

Denote by $I$ the identity matrix of size 2, $X=
\left(\begin{smallmatrix}
  0 & 1\\
1 & 0
\end{smallmatrix}\right)$,
$
Z =
\left(
  \begin{smallmatrix}
1 & 0\\
0 & -1
  \end{smallmatrix}
\right)
$
and $Y = iXZ$.
These operators satisfy the following relations:
$$
\begin{cases}
X^2 = Y^2 = Z^2 = I,\\
XY = -YX = iZ,\\
YZ = -ZY = -iX,\\
ZX = -XZ = iY.\\
\end{cases}
$$
The Pauli group $\tilde \Pauli_1$ for one qubit is the group generated by the matrices:
$$
\tilde \Pauli_1 = \{\pm I, \pm i I, \pm X, \pm iX, \pm Y, \pm iY, \pm Z, \pm i Z\}
$$
Remark that two different non-identity Pauli matrices always anti-commute.

The Pauli group $\tilde \Pauli_n$ on $n$ qubits is the multiplicative group of $n$-fold tensor products of errors of $\tilde \Pauli_1$:
$$
\tilde \Pauli_n = \{ i^a E_1 \otimes E_2 \otimes \dots \otimes E_n \ | \ a = 0, 1, 2 \text{ or } 3 
\text{ and } E_i =I, X, Y \text{ or } Z \}
$$
The complex number $i^a$ is the phase of the Pauli operator.
An important consequence of this construction is the fact that two Pauli errors either commute or anti-commute.
More precisely, given two errors $E$ and $E'$ of $\tilde \Pauli_n$, we have:
$$
E E' = (-1)^{f(E, E')} E' E,
$$
where $f(E, E')$ is the number of components $j$ such that $E_j$ and $E'_j$ are two different non-identity Pauli matrices.
For example, the operators $I \otimes X \otimes Z$ and $X \otimes Y \otimes Z$ in $\tilde \Pauli_3$ anti-commute because they anti-commute only in the second position.
This fact is at the origin of syndrome measurement.

\subsection{Stabilizer codes}

A stabilizer group $S$ is a commutative subgroup of $\tilde \Pauli_n$ which doesn't contain $-I$. 
A stabilizer group is generated by a family of commuting Pauli operators: $S = <S_1, S_2, \dots, S_r>$.
Without loss of generality, we can assume that these operators have phase 1. This ensure us that $-I$ is not in $S$.

Given $S$ a stabilizer group of $\tilde \Pauli_n$, the corresponding
{\em stabilizer code} $C(S)$ is defined as the set of fixed points of
the subgroup $S$ in $\hn$. Using the physical ket notation for vectors, we have:
$$
C(S) = \{ \ket \psi \in \hn \ | \ s \ket \psi = \ket \psi , \forall s \in S\}.
$$
%If you are not familiar with the ket notation, you can omit it. $\ket \psi$ is only a vector.
%A qubit is represented by a vector of $\h$. The space $\hn$ represents all the messages of $n$ qubits.

This subspace is not trivial by construction of a stabilizer group.
The integer $n$ is the length of the quantum code.
Assume that $S$ is generated by $r$ generators: $S = <S_1, S_2, \dots, S_r>$ with $S_i \in \tilde \Pauli_n$.
We call the {\em stabilizer matrix} of $C(S)$  the matrix $\HH \in
\mathcal M_{r, n} (\{I, X, Y, Z\})$ with the $i$-th row representing the generator $S_i$. The coefficient $\HH_{i, j}$ is the $j$-th component of $S_i$. 
For example, the quantum code associated with the 3 commuting generators $S_1 = (I \otimes X \otimes Z \otimes Y \otimes Z), S_2 = (Z \otimes Z \otimes X \otimes I \otimes Z)$, and $S_3 = (I \otimes Y \otimes Y \otimes Y \otimes Z)$ 
is described by the following stabilizer matrix:
$$
\HH =
\begin{pmatrix}
X & Z & I & I & Z\\
Z & X & X & Y & I\\
Y & Y & X & Y & Z\\
\end{pmatrix}
$$
A stabilizer code is completely defined by its stabilizer matrix,
though different stabilizer matrices can define the same group and
therefore the same code.
%A stabilizer code can be completely defined by its stabilizer matrix.
%A stabilizer matrix is a matrix with Pauli coefficients $I, X, Y$ or $Z$ with commuting rows. It defines a stabilizer code.

\subsection{Syndrome of an error}

Given $S = <S_1, S_2, \dots, S_r>$ a stabilizer group of $\tilde \Pauli_n$, 
assume that $\ket \psi \in C(S)$ is subjected to a Pauli error $E \in
\tilde \Pauli_n$. The vector $\ket \psi$ is corrupted to $E \ket \psi$.
To recover the original quantum state, we measure the syndrome to obtain information on the error.
The {\em syndrome} of $E \in \tilde \Pauli_n$ is $\sigma(E) = (\sigma_1, \sigma_2, \dots, \sigma_r) \in \F_2^r$ defined by:
$$
\sigma_i =
\begin{cases}
0 \text{ if } E \text{ and } S_i \text{ commute }\\
1 \text{ if } E \text{ and } S_i \text{ anti-commute }
\end{cases}
$$
Given the corrupted quantum state $E \ket \psi$ 
the syndrome of the error $E$ can be measured.
It satisfies $\sigma(EE') = \sigma(E)+\sigma(E')$. The syndrome of an error $s \in S$ which has no effect on the quantum code is $\sigma(s) = 0$.

\subsection{Minimum distance of a stabilizer code}

The phase $i^a$ of a Pauli error $E \in \tilde \Pauli_n$ does not play
a role because we want to protect 
quantum states of $C(S)$, that is vectors of $C(S)$ defined up to multiplication by a non-zero complex number.
Therefore, we will consider errors $E \in \Pauli_n$ defined up to phases. In
what follows, unless otherwise stated, the Pauli group will be the
abelian quotient group:
$$
\Pauli_n = \tilde \Pauli_n / \{\pm 1, \pm i\}.
$$
Given $E, E'$ in the group $\Pauli_n$, we will say that they commute if they commute in the original group $\tilde \Pauli_n$.
This misuse of language is not problematic because commutation doesn't depend on the phase.

If we receive a quantum state $E \ket \psi$, where $\ket \psi$ is in
the quantum code $C(S)$, we measure its syndrome $\sigma(E)$.
We then apply to $E \ket \psi$ an error $\tilde E$ such that $\sigma(\tilde E) = \sigma(E)$. 
After this process the quantum state is $\tilde E E \ket \psi$. It is corrupted by an error $\tilde E E$ of syndrome 0, because $\sigma(\tilde E E) = \sigma(\tilde E)+\sigma(E) = 0$.
There are two types of error of zero syndrome. If $\tilde E E $ is in $S$, it fixes the quantum code and we have recovered the quantum state.
Otherwise, the original quantum state is probably lost. Errors of zero
syndrome that are not in $S$ are called undetectable or {\em
  problematic} errors.

The above observation leads to the definition of the minimum distance $d$:
$$
d = \min \{ |E| \ | \ E \in \Pauli_n \backslash S, \sigma(E) = 0 \},
$$
where $|E|$ is the weight of $E$, it is the number of non-identity components of $E$.
In other words, $d$ is the minimum weight of a problematic error.
The set of errors of syndrome 0 in $\Pauli_n$ is frequently denoted by
$N(S)$, because it is the normalizer 
and the centralizer of the subgroup $S$ in $\tilde \Pauli_n$.
Thus, the minimum distance is the minimum weight of an error of $N(S) \backslash S$.

\subsection{Degeneracy}

An essential feature of quantum coding theory that sets it appart from
classical coding is degeneracy.
It allows for the same decoding procedure to correct a large number of
different errors.
More precisely, all the errors of a coset $E.S$ can be corrected by the same error $E$.
Indeed, assume that a state $\ket \psi$ of the quantum code is corrupted by an error $Es$, where $s \in S$.
Then, after application of $E$, we recover the original quantum state
because it is fixed by $S$. We have: $EEs \ket \psi = \ket \psi$.
To correct an error $E \in \Pauli_n$, it is sufficient to determine its coset $E.S$.

% \subsection{Dimension of a stabilizer code}

% The dimension of the quantum code $C(S)$ is $2^k$, where $k=n-\rank S$.
% It is the number of encoded qubits.
% Recall that the rank of the abelian group $S$ is the size of a minimum generating family.
% We can also see this rank as the rank of the row space of $\HH$. We will detail this in Section \ref{subsection:rank}. Thus we will denote this rank by $\rank \HH$, where $\HH$ is a stabilizer matrix of the quantum code. Thereby, the dimension formula $k=n-\rank \HH$ coincides with the dimension formula for classical codes.
% For example, the rank of the stabilizer group presented above is 2, because we have: $S_3 = S_1 S_2$. The dimension of the quantum code is then $2^{5-2}=2^3$.

\subsection{The $\F_2$-vector space structure, rank and dimension} 
\label{subsection:rank}
The dimension of the quantum code $C(S)$ is $2^k$, where $k=n-\rank S$.
It is the number of encoded qubits. The quantity $\rank S$ is simply
the rank of the abelian group $S$, i.e. the size of a minimum
generating family. The {\em rate} of the quantum code is defined as $R=k/n$.

The Pauli group $\Pauli_n$ can be seen as an $\F_2$-vector space of dimension $2n$, by the isomorphism:
\begin{align*}
\theta : \Pauli_n & \longrightarrow \F_2^{n} \times \F_2^{n} = \F_2^{2n} \\
X_i &\longmapsto (e_i | 0)\\
Z_i &\longmapsto (0 | e_i)\\
\end{align*}
where $X_i$ is $X$ on the $i$-th component and the identity on the other components.
Errors $Y_i$ and $Z_i$ are defined similarly.
The image of $Y_i = X_i Z_i$ is $\theta(X_i Y_i) = (e_i | e_i)$.
For example, the operator $I \otimes X \otimes Y \otimes Z$ corresponds to the vector $(0110|0011)$.
For this $\F_2$-vector space structure, the addition of vectors in
$\F_2^{2n}$ corresponds to componentwise multiplication of Pauli
errors.

By the isomorphism $\theta$, subgroups of $\Pauli_n$ are sent onto
$\F_2$-linear subspaces of $\F_2^{2n}$. 
The rank of a subgroup of $\Pauli_n$ is therefore also the dimension
of the corresponding subspace. If $\HH$ is a stabilizer matrix we will also write $\rank \HH$ to
denote the rank of its row-space, equivalently the rank of the
associated stabilizer group. Note that we may choose a stabilizer matrix with a larger number $r$ of rows than its rank.

We will find it convenient to keep the notation $I, X, Y$ and $Z$ for
stabilizer matrices, but we stress the binary vector space
structure that we will rely upon heavily in the next section.

With this vector space interpretation, the syndrome application:
\begin{align*}
\sigma : \Pauli_n & \longmapsto \F_2^{r}\\
E & \longrightarrow \sigma(E)
\end{align*}
can be regarded as an $\F_2$-linear map.

\subsection{The CSS construction}

One of the most popular ways of constructing quantum codes is the
Calderbank, Shor and Steane (CSS) construction \cite{CS96,St96}. 
A CSS code is a stabilizer code constructed from a stabilizer group $S$ such that:
$$
S = < S_1, S_2, \dots, S_{r_X}, S_{r_X +1}, \dots, S_{r_X + r_Z} >
$$
where $S_1, S_2, \dots, S_{r_X}$ are included in $\{I, X\}^{\otimes n}$ and $S_{r_X +1}, S_{r_X +2}, \dots, S_{r_X + r_Z}$ belong to $\{I, Z\}^{\otimes n}$.
This simplifies commutation relations because two errors of $\{I, X\}^{\otimes n}$ automatically commute and it is similar in $\{I, Z\}^{\otimes n}$.
In this case the stabilizer matrix $\HH$ is decomposed into two stabilizer matrices $\HH_X$ and $\HH_Z$. The matrix $\HH_X$ is composed of $r_X$ rows representing the stabilizers with coefficients in $\{I, X\}$ and the matrix $\HH_Z$ is composed of $r_Z$ rows which define the stabilizers with coefficients in $\{I, Z\}$.

Remark that the subgroup $\{I, X\}$ is isomorphic to $\F_2$, thus we can write the matrix $\HH_X$ as a binary matrix. The same remark is also valid for $\HH_Z$.
By this last isomorphism, rows of the matrices can be seen as binary vectors of length $n$ and the commutation relation between a row of $\HH_X$ and a row $\HH_Z$ corresponds to the orthogonality of these binary rows in $\F_2^{n}$.

Finally, a CSS code can be defined from two binary matrices $\HH_X
\in \mathcal M_{r_X, n}(\F_2)$ and $\HH_Z \in \mathcal M_{r_Z,
  n}(\F_2)$ with orthogonality between rows of $\HH_X$ and rows of $\HH_Z$ in $\F_2^n$.
The number of thus encoded qubits is:
$$
k=n-\rank \HH_X  - \rank \HH_Z,
$$
because ranks of the binary matrices $\HH_X$ and $\HH_Z$ coincide with ranks of the corresponding groups.
Denote by $C_X$ the classical code 
$\Ker \HH_X = \{c\in \F_2^n, \;\HH_X\, ^t\!c =0\}$ and denote by $C_Z$ the code $\Ker \HH_Z$.
The minimum distance of the quantum code is:
$$
d = \inf \{w(x) \ | \ x \in C_X \backslash C_Z^\perp \cup C_Z \backslash C_X^\perp \},
$$
where $w(x)$ is the Hamming weight of a binary vector.

\medskip

\noindent
{\bf Problematic errors.}
By the isomorphism of section~\ref{subsection:rank}, the error vector
$E$ can be seen as two simultaneous binary vectors,
$\theta(E)=(E_X,E_Z)$. The error $E$ has zero syndrome if and only if
\begin{equation}
  \label{eq:zerosyndrome}
  E_X\in C_Z \;\;\text{and}\;\; E_Z\in C_X
\end{equation}
The error $E$ is problematic if \eqref{eq:zerosyndrome} holds together 
with the condition
\begin{equation}
  \label{eq:problematic}
  E_X\not\in C_X^\perp \;\;\text{or}\;\; E_Z\not\in C_Z^\perp.
\end{equation}

\section{Capacity of the quantum erasure channel} \label{section:stabilizer}

The capacity of a quantum channel is the highest rate of a family of quantum codes with an asymptotic zero error probability after decoding. Such a rate is called achievable.
For the quantum erasure channel of erasure probability $p$, the
capacity $Q$ is $1-2p$ 
when $p \leq 1/2$ and it is zero above
$1/2$. The upper bound 
\begin{equation}
  \label{eq:capacity}
  Q \leq 1-2p
\end{equation}
comes from the no-cloning
theorem, see for example \cite{BDS97}. Therefore, it doesn't rely on
the quantum code structure.
Since our purpose is to obtain improved capacity bounds for particular
families of codes, namely quantum LDPC codes, we need to derive
capacity from the code structure: our first step
is to express achievable rates of stabilizer codes over the quantum erasure channel, as a function of their stabilizer matrices.

\subsection{The quantum erasure channel}

The quantum erasure channel admits several equivalent definitions. 
See for example \cite{GBP97}, \cite{Gr02}, \cite{Pr98}.
As a completely positive trace preserving map, it is given by:
$$
\proj \psi \longmapsto (1-p) \proj \psi + p \proj 2
$$
where $\ket \psi$ is a quantum state in $\h$ and the final state lives in $\C^3 = \h \oplus^\perp \C\ket 2$.
The vector $\ket 2$ is orthogonal to the space $\h$, it corresponds to
a lost qubit.
In this paper, we will use the definition based on the Pauli operators which is well adapted to the stabilizer formalism.
When we use the quantum erasure channel, each qubit is erased independently with probability $p$. An erased qubit is subjected to a random Pauli error $I, X, Y$ or $Z$ with equal probability 1/4 and we know that this qubit is erased.

This description of the quantum erasure channel can be deduced from the definition as a completely positive trace preserving map. Indeed, the orthogonality between $\ket 2$ and $\h$ allows us to measure the erased qubit. After, we replace the lost qubit $\proj 2$ by a totally random qubit of density matrix $I/2$. This random state is the original qubit subjected to a random error $I, X, Y$ or $Z$ with equal probability. Therefore, we recover the second definition.

On $n$ qubits, we denote by $\E \in \F_2^n$, the characteristic vector
of the erased positions. 
Each component of the vector $\E$ follows a Bernoulli distribution of probability $p$. That is, the probability of a given vector $\E \in \F_2^n$ is $p^{|\E|} (1-p)^{n-|\E|}$. The qubit in position $i$ is lost if and only if $\E_i=1$.
In this case, the quantum state is subjected to a random Pauli error
$E \in \Pauli_n$, which act trivially on the non-erased qubits: $E_i =
I$ if $\E_i = 0$. We write this condition $E \subset \E$ and will say
that erasure $\E$ {\em covers} the error $E$. Keep in mind
that $E$ is a Pauli operator with coefficients in $\{I, X, Y, Z\}$ and
$\E$ is a binary vector, the shorthand notation $E \subset \E$ expresses just that
the support of $E$ is included in the set of erased positions. Note
finally that
given $\E \in \F_2^n$, all errors $E \subset \E$ occur with the same probability.

An encoded quantum state $\ket \psi$ is corrupted to a state $E \ket
\psi$ by a random error $E$ for which we have the additional knowledge
$E\subset \E$. 
To recover the original quantum state, we compute the syndrome $\sigma
\in \F_2^r$ and must deduce from the couple $(\E, \sigma)$ an error
${\tilde E} \subset \E$. To correct the effect of $E$ we apply $\tilde
E$ and the final state is $\tilde E E \ket \psi$.
If the errors $E$ and $\tilde E$ are in the same coset modulo $S$,
then $\tilde E E$ is a stabilizer of the quantum code. Thus the final
quantum state is the original state. 
When $\tilde E$ is not equivalent to $E$, we will not, in general, recover the quantum state.
Note that in this case $\tilde E E$ is a problematic error and $\tilde
E E \subset \E$. When this happens, i.e. when the erasure vector covers
a problematic error, we will say that we have a {\em non-correctable}
erasure: otherwise the erasure is correctable.

\medskip
\noindent
{\bf Non-correctable erasures in the CSS case.}
From the characterization \eqref{eq:zerosyndrome} and
\eqref{eq:problematic} of problematic errors, we obtain the simple
characterisation of a non-correctable erasure in the CSS case.

\begin{prop}\label{prop:noncorrectable}
   Let 
$\HH =\left(
  \begin{smallmatrix}
    \HH_X\\
    \HH_Z
  \end{smallmatrix}
\right)$
be the stabilizer matrix of a CSS code, and let $C_X$ and $C_Z$ be the
corresponding classical binary codes. 
The erasure vector $\E\in\F_2^n$ is non-correctable if and only if
there exists a binary vector $v$ whose support is included in the
support of $\E$ and such that
$$v\in C_X\setminus C_Z^\perp \;\;\text{or}\;\; v\in C_Z\setminus C_X^\perp.$$
\end{prop}

\subsection{An example of a non-correctable erasure}

As an example of the general case, consider the stabilizer code defined by the matrix:
$$
\HH =
\begin{pmatrix}
I & X & Z & Y & Z\\
Z & Z & X & I & Z\\
I & Y & Y & Y & Z\\
\end{pmatrix}
$$
If the erasure is $\E = (0, 1, 1, 0, 0)$, there are $2^{2|\E|}=2^4$ possible errors:
$$
\{ E \in \Pauli_n \ | \ E \subset \E \} = 
\{ X_2^{a_2} Z_2^{b_2} X_3^{a_3} Z_3^{b_3} \ | \ a_i, b_i \in \F_2 \},
$$
where the error $X_i$ is the error with $i$-th component $X$ and which is the identity outside $i$.
The operator $Z_i$ is defined similarly and $Y_i$ is the error $X_i Z_i$.

\noindent
Let us focus our attention on the ``erased matrix''
$$
\HH_\E =
\left(
\begin{array}{cc}
X & Z\\
Z & X\\
Y & Y\\
\end{array}
\right)
$$
which is the submatrix of $\HH$ whose columns are the columns indexed by the erased positions.
It is natural to introduce this matrix because the syndrome of an
error included in the erasure $\E$ depends only on 
these columns.
We remark that the third row of this matrix is the product of the
first two rows. Thus, the syndrome $u \in \F_2^3$ of an error $E \subset
\E$ satisfies $u_3 = u_1+u_2$. It depends only on the first two rows of
$\HH_\E$. Therefore, there are $2^2$ different syndrome values for
errors $E$ that are covered by the erasure.

\noindent
Now let us look at the remaining columns.
The non-erased submatrix $\HH_{\bar \E}$ is:
$$
\HH_{\bar \E} =
\left(
\begin{array}{ccc}
I & Y & Z\\
Z & I & Z\\
I & Y & Z\\
\end{array}
\right).
$$
Assume that $E$ and $E'$ are two errors included in $\E$, which are in
the same degeneracy class. That is, they differ 
by right-multiplication by an error $s \in S$. 
The restriction of the error $s = EE'$ to $\bar \E$ is the identity. The rank of this submatrix is $\rank \HH_{\bar \E} = 2$ because the first row and the third row are identical. Therefore, we have two possibilities for $s$: either $s=I^{\otimes 5}$ or $s = S_1 S_3 = I \otimes Z \otimes X \otimes I \otimes I$.
There are two errors in each degeneracy class and four possible
syndrome values: therefore, if there were no problematic error
included in $\E$ the total number of errors included in $\E$ would be
$2\times 4=2^3$, but we have seen that it actually is $2^4$. 
This erasure is not correctable.

\subsection{Two enumeration lemmas}

We will pursue the preceding approach.
Our strategy is to determine the cardinalities of two sets of Pauli errors:
\begin{itemize}
\item $N(S)_\E = \{ E \in N(S) \ | \ E \subset \E \}$,\\
recall that $N(S)$ is the set of Pauli errors of syndrome 0.
\item $S_\E = \{ s \in S \ | \ s \subset \E \}$.
\end{itemize}
We will use the submatrices introduced in the above example.

\bigskip
\textbf{The random submatrix $\HH_\E$:}
Let $\HH$ be a matrix of a stabilizer code. With an erasure $\E \in \F_2^n$, we associate the submatrix $\HH_\E$ of the stabilizer matrix $\HH \in \mathcal M_{r,n}(\Pauli_1)$ composed of the columns of the erased qubits. This is the submatrix of the columns of index $i$ such that $\E_i=1$. Similarly $\HH_{\bar \E}$ is the matrix of the non-erased qubits, corresponding to the conjugate $\bar \E$ of $\E$ defined by: $\bar \E_i = \E_i+1$.
%The rank of a matrix with coefficients in $\Pauli_1$ has been defined in Section \ref{subsection:rank}. It is the rank of the subgroup of $\Pauli_n$ generated by its rows or the dimension of this subgroup regarded as an $\F_2$-vector space.

\begin{lemma}\label{lemma:dimension_syndrome_space}
Let $S$ be a stabilizer group of matrix $\HH \in \mathcal M_{r, n}$.
The set $N(S)_\E$ is an $\F_2$-vector space of dimension $2|\E|-\rank \HH_\E$.
\end{lemma}

\begin{proof}
The $\F_2$-linear structure of the Pauli group $\Pauli_n$ has been
detailed in Section \ref{subsection:rank}. 
The syndrome is an $\F_2$-linear map from $\Pauli_n$ to $\F_2^r$.
Its restriction $\sigma_\E$ to the space of Pauli errors included in
$\E$ is also an $\F_2$-linear map. 
The subspace  $N(S)_\E$ is simply the kernel of $\sigma_\E$. Its dimension is $2|\E| - \dim \im \sigma_\E$.
The restricted syndrome function $\sigma_\E$ depends only on the submatrix $\HH_\E$.
It is straightforward to see that the dimension of its image is the rank of $\HH_\E$.
\end{proof}

\begin{lemma}\label{lemma:dimension_Se}
Let $S$ be a stabilizer group of matrix $\HH \in \mathcal M_{r, n}$.
The set $S_\E$ is an $\F_2$-vector space of dimension $\rank \HH - \rank \HH_{\bar \E}$.
\end{lemma}

\begin{proof}
The set $S_\E$ is the kernel of the $\F_2$-linear map:
\begin{align*} 
S & \longrightarrow \{E \in \Pauli_n \ | \ E \subset \E\}\\
s & \longmapsto s_{|\bar \E}
\end{align*}
By definition of the rank, the image of this application is a space of dimension $\rank \HH_{\bar \E}$, and the group $S$ has dimension $\rank \HH$. Therefore, $\dim S_\E = \rank \HH - \rank \HH_{\bar \E}$.
\end{proof}

\bigskip 
From the lemmas, there are $2^{\rank \HH_\E}$ different syndromes and in each coset modulo $S$
there are $2^{\rank \HH - \rank \HH_{\bar \E} }$ errors included in
$\E$. 
Therefore the number of correctable error patterns is $2^{\rank \HH +
  \rank \HH_\E - \rank \HH_{\bar \E} }$, and since there are $2^{2|\E|}$
error vectors covered by $\E$, the erasure vector $\E$ can be corrected
only if:
\begin{eqnarray}\label{eqn:rank_equation_idea}
2|\E| \leq \rank \HH + \rank \HH_\E - \rank \HH_{\bar \E}.
\end{eqnarray}
The rank of  $\HH$ is $\rank \HH = (1-R)n$ where $R$ is the rate of
the quantum code. When $p\leq 1/2$, there are typically more
non-erased coordinates than erased ones, and it is reasonable to
expect that the larger matrix $\HH_{\bar \E}$ has a higher rank than
the smaller matrix $\HH_\E$. 
Equation~\eqref{eqn:rank_equation_idea} therefore becomes simply $2|\E| \leq \rank \HH$
and the typical weight of an erasure being $|\E| = np$, we obtain:
$$
R \leq 1-2p
$$
which recovers \eqref{eq:capacity} for the class of stabilizer codes.
In the next section we will make this informal argument rigorous and
pave the way for improvements for particular classes of quantum codes.

\subsection{A combinatorial bound on the capacity}

Now, we will give a rigorous proof using an entropic formulation of this idea and Fano's inequality.

Recall that a rate $R \in [0, 1]$ is achievable if there exists a
family of codes of rates $R_t$ converging to $R$ with vanishing error probability after decoding.
Denote expectation by $\mathbb E$.
Let $(\HH_t)_t$ be a sequence of stabilizer matrices and denote by $n_t$ the length of the code defined by the matrix $\HH_t$.

\begin{defi}\label{defi:D}
The rank difference function $D$ associated with the sequence $(\HH_t)_t$ of stabilizer matrices is $D(p) = \limsup_t \Delta_t(p)$ where :
$$
\Delta_t(p) = \frac{ \mathbb E_p [\rank \HH_{t, \bar \E} - \rank \HH_{t, \E}] }{n_t}.
$$
\end{defi}

\begin{theo}\label{theo:capacity_stabilizer}
Achievable rates of a sequence of stabilizer codes of matrices $(\HH_t)_{t \in \N}$, over the quantum erasure channel of probability $p$, satisfy:
$$
R \leq 1 - 2p - D(p).
$$
where $D(p)$ is the rank difference function of the family $(\HH_t)_t$.
\end{theo}

\begin{proof}
We shall apply the classical Fano inequality, see for example
\cite{CT91}.
Recall that if $X, Y$ and $\hat X$ are any three random variables such
that $\hat X$ depends only on $Y$, 
then Fano's inequality states:
$$
P_{err} := \mathbb P(\hat X \neq X) \geq \frac{H(X|Y) - 1}{\log(|\X|)}.
$$
where $X$ takes its values in $\X$.

Over the quantum erasure channel, $\E$ is the erasure vector random variable with distribution
$\mathbb P(\E=v)= p^{|v|}(1-p)^{n-|v|}$. The error random variable $E$ is uniformly distributed among the errors acting on the erased components.
We apply Fano's inequality when $X$ is the information we need to
recover the quantum state, namely the coset $E.S$ of the Pauli error
vector $E$. We set the variable $Y$ to be the couple $Y=(\E , \Sigma)$
where $\E$ is the erasure vector random variable defined above and $\Sigma =\sigma(E)$ is the syndrome
of $E$. The variable $\hat X$ is the best possible estimation of $X$
given $Y$, meaning here that the decoding error probability $P_{err}$ is the probability to have $X \neq \hat X$.

The conditional entropy is decomposed as:
$$
H(X | \E, \Sigma) = \sum_{v, y} \mathbb P \big( (\E, \Sigma) = (v, y) \big) H(X | \E = v, \Sigma = y).
$$
We have:
$$
H(X | \E = v, \Sigma = y) = 2|v|-\rank \HH + \rank \HH_{\bar v}-\rank \HH_v,
$$
the proof of which is detailed in Lemma~\ref{lemma:conditional_entropy} below.
We see that the value of $H(X | \E = v, \Sigma = y)$ is independent of $y$, thus we have:
\begin{align*}
H(X | \E, \Sigma)
& =
\sum_v \mathbb P (\E = v) \big( 2|v| - \rank \HH - \rank \HH_{v} + \rank \HH_{\bar v} \big)\\
& = 2np - \rank \HH + \mathbb E_p (\rank \HH_{\bar \E} - \rank \HH_{\E}).
\end{align*}

The random variable $X = E.S$ takes on values in the quotient group
$\X = \Pauli_n/S$. This quotient group is composed of $|\X| = 2^{2n-\rank \HH}$ classes.
From Fano's inequality we get, upperbounding the denominator by $2n-\rank \HH \leq 2n$,
\begin{align*}
P_{err}  &
\geq \frac{2np - \rank \HH + \mathbb E_p (\rank \HH_{\bar \E}-\rank \HH_{\E}) - 1}{2n}
\end{align*}

\noindent
The rate of the quantum code is $R = 1- \rank \HH /n$.
If the error probability goes to zero, then the rate of the quantum
code family satisfies:
$$
\limsup R \leq 1- 2p - D(p).
$$
\end{proof}

\begin{lemma} \label{lemma:conditional_entropy}
Let $S$ be a stabilizer group of matrix $\HH$.
The conditional entropy of $X = E.S$ given $\E=v$ and $\Sigma=y$ is:
$$
H(X | \E = v, \Sigma = y) = 2|v|-\rank \HH + \rank \HH_{\bar v}-\rank \HH_v,
$$
when the probability to have $\E = v$ and $\Sigma = y$ is non zero. 
\end{lemma}

\begin{proof}
Recall that given the erasure $\E$, the distribution of $E$ is
uniform inside the support of~$\E$. Therefore,
the probability of a coset $E.S$, assuming that the erasure is $\E=v$ and the syndrome is $\Sigma=y$, is:
$$
\mathbb P( X = E.S | \E=v, \Sigma=y) = 
\frac{ |\{ P \in E.S \ | \ P \subset v, \sigma(P) = y \}| }
{ |\{ P \in \Pauli_n \ | \ P \subset v, \sigma(P) = y \}| }.
$$
When this probability is non-zero, 
by linearity, (multiply by an operator $T\subset v$ of syndrome
$y$), we can assume that $y=0$. The set in the denominator
is the subgroup $N(S)_v$ and
the set in the numerator is a coset of the subgroup $S_v$. Applying
Lemmas \ref{lemma:dimension_syndrome_space} and \ref{lemma:dimension_Se} we have:
$$
\mathbb P( X = E.S | \E=v, \Sigma=y) = \frac{|S_v|}{|N(S)_v|}
= 2^{-2|v|+\rank \HH  -\rank \HH_{\bar v}+\rank \HH_v }.
$$
whence 
$$
H(X | \E = v, \Sigma = y) = 2|v|-\rank \HH + \rank \HH_{\bar v}-\rank \HH_v.
$$
\end{proof}

\bigskip
The following corollary proves the efficiency of our method. We recover the upper bound given by the capacity of the quantum erasure channel \cite{BDS97}. This bound is deduced only from combinatorial properties of stabilizer codes. It doesn't involve the no-cloning theorem.

\begin{coro}\label{cor:1-2p}
Achievable rates of a sequence of stabilizer codes of matrices $(\HH_t)_{t \in \N}$, over the quantum erasure channel of probability $p$, satisfy:
$$
R \leq 1 - 2p,
$$
when $p \leq 1/2$.
\end{coro}

\begin{proof}
To prove the corollary, it suffices to remark that $\Delta_t(p)$ in
Theorem~\ref{theo:capacity_stabilizer} is
non-negative when $p \leq 1/2$. Observe that we can
write the function $\Delta_t$ as:
$$
\Delta_t(p) = \phi_t(1-p) - \phi_t(p),
$$
where $\phi_t(p) = \mathbb E_p (\rank \HH_{t, \E})/n_t$. It is intuitively
clear and it is formally stated and proved in Appendix~\ref{appendix:concavity_inequality}
(Proposition~\ref{prop:increasing}) that $\phi_t$ is an increasing
function of $p$. The corollary follows.
\end{proof}

Our goal is now to improve on Corollary~\ref{cor:1-2p} by finding non-zero
lower bounds on $D(p)$. This cannot be done for stabilizer codes in
general since they are known to achieve capacity of the quantum
erasure channel, but we can obtain such improvements for {\em sparse}
quantum codes, i.e. codes that have sparse stabilizer matrices.
Our most general result in this direction is
Theorem~\ref{theo:stabilizer} below. It is
somewhat reminiscent of an upper bound on achievable rates of
classical LDPC codes for the classical erasure channel \cite{RU08}.

\subsection{Reduction to the study of the mean rank of a random
  submatrix of the stabilizer matrix}\label{section:method}

\begin{theo}
  \label{theo:stabilizer}
  Let $\CC$ be any family of stabilizer codes of rates at least $R$
  and achieving vanishing decoding error probability over the quantum
  erasure channel of erasure probability $p$. Suppose furthermore that
  every code $C\in\CC$ has a set of generators of its stabilizer group
  whose weights are all upper bounded by $m$. Then we have:
  $$R \leq (1-2p)\frac{1-(1-p)^{m-1}}{1-(1-2p)(1-p)^{m-1}}.$$
\end{theo}

{\bf Method:}
Let $\HH$ be a stabilizer matrix and set $\phi(p)=\mathbb E_p (\rank \HH_\E)/n$. Denote by $\Delta$ the function $\Delta(p)=\phi(1-p)-\phi(p)$.
To apply Theorem~\ref{theo:capacity_stabilizer} we need a lower bound
on $\Delta(p)$. Because the function $\phi$ is concave, any upper bound
$\phi(p)\leq M(p)$ implies the lower bound on $\Delta$:
\begin{equation}
  \label{eq:fconcave}
  \Delta(p) \geq \frac{1-2p}{1-p}\left(\frac {\rank \HH} n - M(p)\right).
\end{equation}
The formal proof of the concavity of $\phi$ and of \eqref{eq:fconcave} is somewhat technical and not
necessary to the understanding of the main ideas, therefore it is
placed in Appendix~\ref{appendix:concavity_inequality} (Proposition \ref{prop:bound_g}).

\begin{proof}[Proof of Theorem~\ref{theo:stabilizer}]
  Let $\HH$ be a stabilizer matrix of a code $C\in\CC$. Set
  $\phi(p)=\mathbb E_p (\rank \HH_{\E})/n$. Let $\hh_\E^0$ stand for the
  number of zero rows in the random submatrix $\HH_\E$. We have:
  \begin{align*}
    \phi(p) \leq &\frac 1n \left[\rank \HH - \mathbb E_p (\hh_{\E}^0)\right]\\
         \leq & \frac {\rank\HH}n - \frac {\rank \HH}n (1-p)^m.
  \end{align*}
  Applying \eqref{eq:fconcave} we get:
  $$\Delta(p)\geq \frac{1-2p}{1-p}\frac {\rank \HH}n (1-p)^m \geq  \frac{1-2p}{1-p}(1-R) (1-p)^m.$$
  From Theorem~\ref{theo:capacity_stabilizer} we now have:
  $$R\leq 1-2p - (1-R)(1-2p)(1-p)^{m-1}$$
  and the result follows after rearranging.
\end{proof}

As an example let us consider the family of {\em color codes}
\cite{BM06}.
Color codes are defined from a trivalent tiling of a surface by faces
and the associated stabilizer matrices have rows of weight bounded by
the maximum length (number of edges) of a face. Hence
Theorem~\ref{theo:stabilizer} applies to this family that cannot be
capacity-achieving if the faces stay with bounded length.

\section{Achievable rates of $(2, m)$ CSS codes} \label{section:CSS(2,m)}

We now turn to deriving a refined upper bound on achievable rates of a particular family
of quantum LDPC codes. Let us say that a binary matrix is of type
$(2,m)$ if every one of its rows is of weight $m$ and every column is
of weight $2$. We shall say that a quantum code is a
$(2,m)$ CSS code if its stabilizer matrix $\HH$
decomposes in two matrices $\HH_X$ and $\HH_Z$, {\em each of which is
a $(2, m)$ matrix}.

\subsection{The $2$-complex associated to a $(2, m)$ CSS code}
\label{section:2complex}
The matrix $\HH_X$, viewed as a binary matrix, can be seen as the
incidence matrix of a finite graph $G_X$. The vertex set $V$ of the
graph $G_X$ is defined as the set of rows of $\HH_X$, and two vertices
$i$ and $i'$ are declared to be incident if there is column $j$ such
that there are $1$'s in positions $(i,j)$ and $(i',j)$. The Edge set
of the graph $G_X$ can therefore be indexed by the columns of $\HH_X$.
The constant row weight $m$ of $\HH_X$ means that the graph $G_X$ is
regular
(every vertex has $m$ neighbours).

Recall that the classical code $C_X$ is the set of vectors of $\F_2^n$
orthogonal to the rows of $\HH_X$. The code $C_X$ is generated by the
vectors whose supports coincide with cycles of the graph $G_X$
(actually $C_X$ is {\em exactly} the set of cycles of $G_X$ when one
allows cycles to be non-connected subgraphs) and $C_X$ is classically
called the {\em cycle code} of the associated graph $G_X$.

If the graph $G_X$ has $n$ edges, the dimension of the code $C_X$ is given by: 

\begin{lemma}\label{lemma:dimCX}
  We have $\dim C_X = n-|V| +\kappa_X$ where $\kappa_X$ is the number of
  connected components of $G_X$.
\end{lemma}
This is a classical result, see for example \cite{Be73}.

Since the rows of $\HH_Z$ are orthogonal to the rows of $\HH_X$, the
supports of the rows of $\HH_Z$ are cycles of the graph $G_X$. The
graph $G_X$ together with the set of supports of the rows of $\HH_Z$
is called a $2$-complex, and the supports of the rows of $\HH_Z$ are
particular cycles that are called {\em faces}. That $\HH_Z$ is a
$(2,m)$ matrix means that every edge is incident to exactly two faces.

In the same way that the matrix $\HH_X$ defines a graph $G_X$, the
matrix $\HH_Z$ defines a graph $G_Z$, which together with $\HH_X$ also
makes up a $2$-complex. The two complexes are said to be dual to each
other: faces are vertices of the dual complex.

We shall say that the $(2,m)$ CSS code (or equivalently the associated
$2$-complex) is {\em proper} if the two graphs $G_X$ and $G_Z$ are
connected and have girth (smallest cycle size) equal to $m$.

It is not immediate that proper $(2,m)$ CSS codes even exist, and the
existence of families of $(2,m)$ CSS codes with growing
minimum distance is even less obvious. One way of coming up with
such families is through the construction of combinatorial surfaces.
the associated $(2,m)$ CSS codes are a highly regular instance of {\em surface
codes} \cite{BM073,Ze09}. 
An example of such a surface is given on Figure~\ref{fig:G(5)mod2} for
$m=5$. 
\begin{figure}
  \centering
  \begin{tikzpicture}[scale=2.5]

\tikzstyle{every node}=[circle, draw, fill=black!100, inner sep=0pt, minimum width=2pt]

%pentagone central
\draw
(0:1.3) \foreach \x in {72,144,...,359} { -- (\x:1.3) } -- cycle
(0:1.3) node [label = left:$5$] (a) {}
(72:1.3) node [label = below:$2$](b) {}
(144:1.3) node [label = below right:$0$](c) {}
(216:1.3) node [label = above right:$3$](d) {}
(288:1.3) node [label = above:$9$](e) {};

%on ajoutes les aretes autour des sommets du pentagone
\draw[shift=(a)]
(0:.9) node [label = above left:$10$](a2) {}
(60:.9) node (a3) [label = left:$1$]{}
(300:.9) node (a1) [label = left:$15$]{};
\draw
(a1)--(a)--(a2)
(a)--(a3);

\draw[shift=(b)]
(72:.9) node [label = left:$13$] (b2) {}
(132:.9) node [label = below:$7$] (b3) {}
(12:.9) node [label = below:$14$](b1) {};
\draw
(b1)--(b)--(b2)
(b)--(b3);

\draw[shift=(c)]
(144:.9) node [label = below:$8$](c2) {}
(204:.9) node [label = below right:$6$](c3) {}
(84:.9) node [label = below right:$1$](c1) {};
\draw
(c1)--(c)--(c2)
(c)--(c3);
\draw[color =blue, line width=1.5pt, ->]
(c)--(c1);
%\draw[shift =(c3), color=red, fill=red!20]
%(110:0.1cm) arc (0:90:.1cm);

\draw[shift=(d)]
(216:.9) node [label = right:$11$](d2) {}
(276:.9) node [label = above right:$15$](d3) {}
(156:.9) node [label = above right:$7$](d1) {};
\draw
(d1)--(d)--(d2)
(d)--(d3);

\draw[shift=(e)]
(288:.9) node [label = above right:$4$](e2) {}
(348:.9) node [label = above:$14$](e3) {}
(228:.9) node [label = above:$6$](e1) {};
\draw
(e1)--(e)--(e2)
(e)--(e3);

%fermeture des faces qui contiennent une arete du pentagone
\draw
(36:2.1) node [label = below left:$11$](f){}
(108:2.1) node [label = below:$4$](g){}
(180:2.1) node [label = right:$10$](h){}
(252:2.1) node [label = above right:$13$](i){}
(324:2.1) node [label = above left:$8$](j){};
\draw
(a3)--(f)--(b1)
(b3)--(g)--(c1)
(c3)--(h)--(d1)
(d3)--(i)--(e1)
(e3)--(j)--(a1);

%fermeture des faces qui on un sommet en commun avec le pentagone
\draw[shift=(a2)]
(70:.4) node [label = left:$8$](a24){};
\draw[shift=(a3)]
(345:.4) node [label = below left:$0$](a31){};
\draw
(a2)--(a24)--(a31)--(a3);
\draw[color=blue, line width=1.5, ->]
(a31)--(a3);

\draw[shift=(b2)]
(142:.4) node [label = below:$11$](b24){};
\draw[shift=(b3)]
(57:.4) node [label = below right:$3$](b31){};
\draw
(b2)--(b24)--(b31)--(b3);

\draw[shift=(c2)]
(214:.4) node [label = below right:$4$](c24){};
\draw[shift=(c3)]
(129:.4) node [label = right:$9$](c31){};
\draw
(c2)--(c24)--(c31)--(c3);

\draw[shift=(d2)]
(286:.4) node [label = above right:$10$](d24){};
\draw[shift=(d3)]
(201:.4) node [label = above:$5$](d31){};
\draw
(d2)--(d24)--(d31)--(d3);

\draw[shift=(e2)]
(358:.4) node [label = above:$13$](e24){};
\draw[shift=(e3)]
(273:.4) node [label = above left:$2$](e31){};
\draw
(e2)--(e24)--(e31)--(e3);

%ajout des 5 faces totalement exterieures :
\draw[shift=(j)]
(5:.3) node [label = below:$10$](j3){};
\draw[shift=(j3)]
(45:.2) node [label = right:$6$](j34){};
\draw[shift=(a1)]
(345:.3) node [label = above:$12$](a11){};
\draw
(j)--(j3)--(j34)--(a11)--(a1);

\draw[shift=(f)]
(77:.3) node [label = right:$13$](f3){};
\draw[shift=(f3)]
(117:.2) node [label = above:$15$](f34){};
\draw[shift=(b1)]
(57:.3) node [label = above left:$12$](b11){};
\draw
(f)--(f3)--(f34)--(b11)--(b1);

\draw[shift=(g)]
(149:.3) node [label = above:$8$](g3){};
\draw[shift=(g3)]
(189:.2) node [label = above left:$14$](g34){};
\draw[shift=(c1)]
(129:.3) node [label = left:$12$](c11){};
\draw
(g)--(g3)--(g34)--(c11)--(c1);

\draw[shift=(h)]
(221:.3) node [label = above left:$11$](h3){};
\draw[shift=(h3)]
(261:.2) node [label = left:$1$](h34){};
\draw[shift=(d1)]
(201:.3) node [label = below:$12$](d11){};
\draw
(h)--(h3)--(h34)--(d11)--(d1);

\draw[shift=(i)]
(293:.3) node [label = left:$4$](i3){};
\draw[shift=(i3)]
(333:.2) node [label = below:$7$](i34){};
\draw[shift=(e1)]
(273:.3) node [label = right:$12$](e11){};
\draw
(i)--(i3)--(i34)--(e11)--(e1);
  \end{tikzpicture}
  \caption{An example of self-dual 5-regular tiling of a surface of
    genus 4 composed of 16 vertices, 40 edges and 16 faces. Each face
    is represented once. Vertices are represented several times to
    allow this planar representation.
    Each boundary edge is represented twice. multiple replicas of
    vertices and edges are identified to create the surface.
    In bold the identification of the edges $\{0, 1\}$.}
  \label{fig:G(5)mod2}
\end{figure}
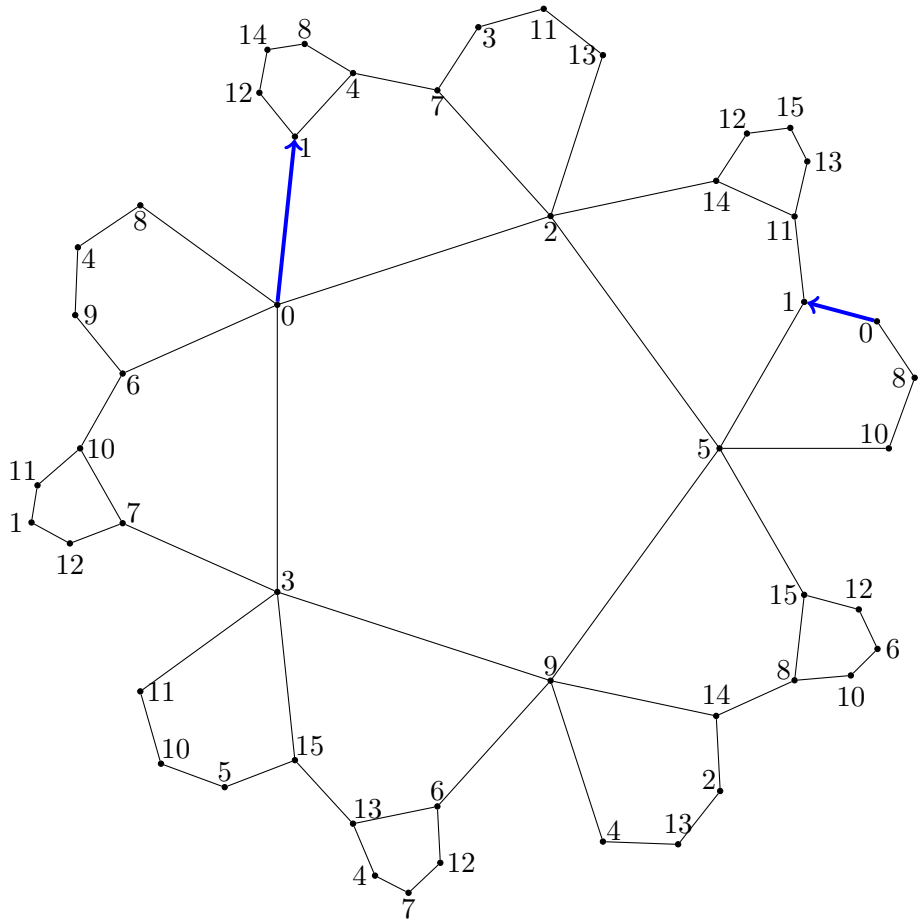
The
associated $(2,5)$ matrices $\HH_X$ and $\HH_Z$ are given on 
Figure~\ref{fig:5_2_code}. It is the smallest $(2,5)$ CSS code we have
found such that the associated graphs $G_X$ and $G_Z$ are both
connected and simple (without multiple edges). The only method we know
of that allows the construction of proper $(2,m)$ CSS codes involves
sophisticated number-theoretic arguments and combinatorial surfaces.
We shall take up this matter
in Section~\ref{section:percolation} where upper bounds on the
achievable rate of quantum $(2,m)$ CSS codes lead to upper bounds on the
critical probability for associated families of infinite tilings.

\begin{figure}
$$
\HH_X =
\begin{pmatrix}
0 & 1 & 2 & 3 & 8\\
1 & 4 & 5 & 11 & 20\\
2 & 6 & 7 & 14 & 25\\
0 & 9 & 10 & 18 & 28\\
5 & 12 & 13 & 22 & 32\\
4 & 7 & 15 & 21 & 31\\
3 & 16 & 17 & 27 & 36\\
6 & 10 & 13 & 19 & 23\\
8 & 12 & 24 & 33 & 38\\
9 & 15 & 17 & 22 & 26\\
16 & 19 & 21 & 24 & 29\\
11 & 28 & 29 & 30 & 35\\
20 & 23 & 27 & 34 & 39\\
14 & 32 & 35 & 36 & 37\\
25 & 26 & 30 & 33 & 34\\
18 & 31 & 37 & 38 & 39\\
\end{pmatrix}
\qquad
\HH_Z =
\begin{pmatrix}
0 & 2 & 7 & 9 & 15\\
1 & 2 & 5 & 6 & 13\\
0 & 3 & 10 & 16 & 19\\
1 & 4 & 8 & 21 & 24\\
3 & 8 & 12 & 17 & 22\\
4 & 7 & 11 & 25 & 30\\
5 & 12 & 20 & 33 & 34\\
6 & 10 & 14 & 28 & 35\\
9 & 17 & 18 & 36 & 37\\
11 & 19 & 20 & 23 & 29\\
13 & 23 & 27 & 32 & 36\\
14 & 22 & 25 & 26 & 32\\
15 & 26 & 31 & 33 & 38\\
16 & 24 & 27 & 38 & 39\\
18 & 21 & 28 & 29 & 31\\
30 & 34 & 35 & 37 & 39\\
\end{pmatrix}
$$
\caption{Two matrices of size $16 \times 40$ defining a $(2, 5)$ CSS
  code of parameters $[[40, 10, 4]]$. Columns are indexed by the
  integers $\{0, 2, \dots, 39\}$ and rows are described by their supports.
This code is the surface code defined from the $2$-complex of Figure \ref{fig:G(5)mod2}.}
\label{fig:5_2_code}
\end{figure}

\subsection{Reduction to the study of the mean number of connected components
in the subgraph of a graph}

Recall from our proof method described in Section~\ref{section:method}
that our objective is to find an upper bound on the function $\phi(p)=
\mathbb E_p (\rank \HH_\E)/n$ for a stabilizer matrix $\HH$.
Since we are dealing with the stabilizer matrix 
$\HH = \left(
\begin{smallmatrix}
  \HH_X \\ \HH_Z
\end{smallmatrix}
\right)$
of a CSS code, we have $\phi(p) = \phi_X(p) + \phi_Z(p)$ where
$\phi_X(p)= \mathbb E_p (\rank \HH_{X,\E})/n$ and $\phi_Z(p)= \mathbb
E_p (\rank \HH_{Z,\E})/n$. 
Recall that the two matrices $\HH_X$ and $\HH_Z$ are $(2,m)$ matrices.
Our problem is therefore to bound from above the mean rank of random
submatrix of a fixed binary $(2,m)$ matrix.

In the rest of this section, let $H \in \mathcal M_{r, n}$ therefore stand for a binary
matrix of type $(2,m)$. We create the submatrix $H_\E$ by keeping
every column of $H$ with probability $p$ and independently of the others.

\bigskip
\textbf{The random subgraph $G_\E$:}
As described in section~\ref{section:2complex} we can regard the
matrix $H$ as the incidence matrix of a graph $G$ with vertex set $V$.

Given the random vector $\E \in \F_2^n$, we denote by $G_\E$ the subgraph of $G$ of
incidence matrix $H_\E$. Assume that each component of $\E$ is 1 with
probability $p$ and 0 with probability $1-p$, independently. Then, the
graph $G_\E$ is the random subgraph of $G$ with unchanged vertex set,
and obtained by taking each edge,
independently, with probability $p$. In other words,
taking a random submatrix $H_\E$ of $H$ corresponds to taking a random
subgraph $G_\E$ of $G$. For the random subgraph $G_\E$,  Lemma~\ref{lemma:dimCX} translates into:

\begin{lemma}\label{lemma:rank_subgraph}
If $H$ is a binary $(2,m)$ matrix with $|V|$ rows, then $|V| -
\rank H_\E$ is equal to the number $\kappa_\E$ of connected components 
of the subgraph $G_\E$. That is:
$$
\rank H_\E = |V| - \kappa_\E.
$$
\end{lemma}

\subsection{Bound on the mean number of connected components in the graph} \label{subsection:graph_components}

From Lemma \ref{lemma:rank_subgraph}, the mean rank of the submatrix $H_\E$ can be expressed as a function of the expected number of connected components of the graph $G_\E$:
\begin{equation}
  \label{eq:f(p)}
  \frac{\Esp_p(\rank H_\E)}{n} = \frac{|V|}{n} - \frac{\mathbb E_p (\kappa_\E)}{n}
\end{equation}

If we upper bound $\Esp_p(\rank H_\E)/n$ by writing that $\mathbb E_p (\kappa_\E)$ is
lower bounded by the expected number of isolated vertices, we will
simply recover Theorem~\ref{theo:stabilizer}. We will proceed to
derive a more precise lower bound by enumerating larger connected
components.

\bigskip
Assuming that
the $m$-regular graph $G$ constructed from the matrix $H$ has no small
cycles, it looks like the $m$-regular tree $G_m$ in any sufficiently small neighbourhood.
Thus, for our enumeration problem, we introduce the number $a_k$ of
subtrees of $G_m$, with $k$ edges, 
containing a fixed vertex $x$ of $G_m$.
We will use a generating function approach: for background,
two classical references on this subject are \cite{FS09} and \cite{Wi06}.

\textbf{The generating function for rooted trees:}
Let $G_m$ be the $m$-regular tree and let $x$ be a fixed vertex of
$G_m$ called a root.
Let us define the generating function for rooted tree of degree $m$ as the real function:
$$
T_m(z) = \sum_{k \geq 0} a_k z^k,
$$
where $a_k$ is the number of subtrees of $G_m$, with $k$ edges, containing the root $x$.
To compute this generating function, it is useful to introduce the auxiliary generating function.
$$
T^1_m(z) = \sum_{k \geq 0} b_k z^k,
$$
where $b_k$ is the number of subtrees $\mathcal T$ of $G_m$ such that :
\begin{itemize}
\item $\mathcal T$ is composed of $k$ edges,
\item the root $x$ is included in $\mathcal T$,
\item $\mathcal T$ contains a fixed edge $\{x, y\}$ among edges
  incident to $x$, and no other edge incident to $x$ is contained in $\mathcal T$.
\end{itemize}

This generating function doesn't depend on the particular
choice of the edge ${x, y}$, by regularity of $G_m$.
The function $T_1(z)$ is sometimes called the generating function of
planted rooted subtrees, since $x$ has degree one in such a subtree.

Coefficients $b_k$ can be computed easily using the Lagrange inversion
Theorem \cite{FS09} because $T^1_m$ satisfies:
$$
T^1_m(z) = z (1 + T^1_m(z))^{m-1}
$$
This formula comes from the fact that every vertex of the tree except the root $x$ has $m-1$ sons.
We get:
$$
b_k = \frac{1}{k}\binom{k(m-1)}{k-1}.
$$
Then, the computation of the $a_k$ follows from the expression of $T_m$ as a function of $T^1_m$.
$$
T_m(z) = (1+T^1_m(z))^m.
$$
To see this formula remark that a subgraph of $G_m$ containing the
root $x$ can be decomposed into at $m$ planted subtrees of root $x$.
This method allows the computation of a large number of coefficients $a_k$ using symbolic computation software.

We can now state an upper bound on the expected rank of the submatrix $H_\E$ involving the numbers $a_k$.
\begin{prop}\label{prop:bound_fn_LDPC2}
If $H$ is a binary $(2, m)$ matrix whose associated graph $G$ has
girth (smallest cycle size) at least $\delta+2$, then
$$
\frac{\Esp_p(\rank H_\E)}{n} \leq \frac{2}{m} \big( 1 - (1-p)^m S_\delta( p(1-p)^{m-2} ) \big),
$$
where $S_\delta(z) = \sum_{k=0}^\delta \frac{a_k}{k+1} z^k$ and the
$a_k$ are the coefficients of the generating function $T_m$.
\end{prop}

\begin{proof}
From \eqref{eq:f(p)} we want a lower bound on the expected number of connected components in the graph $G_\E$.

The graph $G_\E$ is constructed from the edge set of $G$ by choosing each edge, independently, with probability $p$.
Let us compute the expected number of connected components with 0
edges. This is the average number of isolated points in the random subgraph $G_\E$.
Denote by $X_0$ the random variable which associate with a random vector $\E$, the number of isolated points in $G_\E$. We can write $X_0(\E) = \sum X_v(\E)$ where: $X_v(\E) = 1$ if the vertex $v$ is isolated in $G_\E$ and $X_v(\E) = 0$ otherwise.
By linearity of expectation, we have:
$$
\mathbb E (X_0) = \sum_v \mathbb E (X_v) = \sum_v \mathbb P(X_v = 1) = |V|(1-p)^m.
$$
Indeed, each vertex is bordered by $m$ edges, therefore $\mathbb P(X_v = 1) = (1-p)^m$, for every vertex~$v$.

This idea can be used with components of size (number of edges) $k
\leq \delta$.
For $C$ a $k$-edge connected subgraph of $G$, let $X_C$ denote the
random variable equal to $1$ if $C$ is a connected component of the
random graph $G_\E$, and $0$ otherwise. 
The average number of connected components of size $k$ is:
$$
\mathbb E (X_k) = \sum_{\substack{C \text{ connected} \\ \text{ subgraph of size } k}} \mathbb E (X_C) = \frac{|V|}{k+1} a_k (1-p)^m (p(1-p)^{m-2})^k.
$$
To prove the second equality we use two lemmas proved below.
From Lemma \ref{lemma:subgraph_probability}, the expected value of $X_C$ is $(1-p)^m(p(1-p)^{m-2})^k$, independently of the subgraph $C$ with $k$ edges. Lemma~\ref{lemma:subgraph_number} guarantees that the number of connected subgraph $C$ is $\frac{|V|}{k+1} a_k$.
Finally, the quotient $\frac{|V|}{n}$ is exactly $\frac{2}{m}$.
This proves the proposition.
\end{proof}

\begin{figure}[h]
\centering
\begin{tikzpicture}[scale=0.5]

\tikzstyle{every node}=[circle, draw, fill, inner sep=0pt, minimum width=4pt]

\draw
(0,0) node(a){}
(-4,-2) node(b1){}
(0,-2) node(b2){}
(4,-2) node(b3){};

\draw
(-5,-4) node(c1){}
(-3,-4) node(c2){}
(-1,-4) node(c3){}
(1,-4) node(c4){}
(3,-4) node(c5){}
(5,-4) node(c6){};

\draw
(-5.5,-6) node(d1){}
(-4.5,-6) node(d2){}
(-3.5,-6) node(d3){}
(-2.5,-6) node(d4){}
(-1.5,-6) node(d5){}
(-0.5,-6) node(d6){}
(0.5,-6) node(d7){}
(1.5,-6) node(d8){}
(2.5,-6) node(d9){}
(3.5,-6) node(d10){}
(4.5,-6) node(d11){}
(5.5,-6) node(d12){};

\draw
(a)--(b1)
(a)--(b2)
(a)--(b3);

\draw
(c1)--(b1)--(c2)
(c3)--(b2)--(c4)
(c5)--(b3)--(c6);

\draw
(d1)--(c1)--(d2)
(d3)--(c2)--(d4)
(d5)--(c3)--(d6)
(d7)--(c4)--(d8)
(d9)--(c5)--(d10)
(d11)--(c6)--(d12);
\end{tikzpicture}

\caption{A ball of radius 3 in a 3-regular graph with girth $\geq 7$}
\label{fig:ball_graph}
\end{figure}
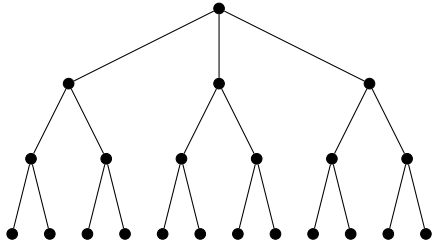

\begin{lemma}\label{lemma:subgraph_probability}
Let $G$ be a $m$-regular graph of girth at least $\delta+2$.
If $C$ is a connected subgraph of $G$ with $k \leq \delta$ edges, then $C$ is a connected component of the random graph $G_\E$ with probability $(1-p)^m(p(1-p)^{m-2})^k$.
\end{lemma}

\begin{proof}
If $k=0$, then $C$ is an isolated point. This component appears in the random graph $G_\E$ with probability $(1-p)^m$ by $m$-regularity of $G$.

Assume the formula true for every connected subgraph $C$ of size $k-1$, with $k \leq \delta$.
A subgraph $C$ of $G$ of size $k-1$ is included in a ball of radius $k-1$.
We will prove that the formula remains true if we add an edge to $C$.
Let $x$ be a vertex of $C$ and let $\{x,y\}$ be an edge which is not in $C$.
Consider the graph $C'= C \cup \{x,y\}$. It contains $k$ edges.
%What is the number of neighbors edges of $C'$?
Denote by $\partial C$ the set of edges which have exactly one endpoint in $C$.
Similarly, $\partial C'$ is the first neighbourhood of $C'$.
The set $\partial C'$ contains $\partial C$ except $\{x,y\}$. Moreover it contains $m-1$ new edges: the edges $\{y,z\}$ for $z \neq x$. These edges were not in $\partial C$. Indeed, if $\{y,z\}$ is included in $\partial C$, define the graph $C'' = C' \cup \{y,z\}$. It contains $k+1$ edges and it covers a cycle, because without the edge $\{y,z\}$ it is still connected. 
This is impossible because the shortest cycle has length at least $\delta+2$.
Thus the formula is satisfied for all $k \leq \delta$.
\end{proof}

\begin{lemma}\label{lemma:subgraph_number}
Let $G$ be an $m$-regular graph of girth at least $\delta+2$.
The number of connected subgraphs of $G$ with $k \leq \delta$ edges is at least:
$$
\frac{|V|}{k+1} a_k.
$$
where $a_k$ are the coefficient of the generating function $T_m$.
\end{lemma}

\begin{proof}
Connected subgraphs with $k$ edges are trees and contain $k+1$
vertices. It should be clear that given any vertex $x$, the number of
ways of constructing a $k$-edge subgraph containing $x$, and hence the
number of such subgraphs, is the same in $G$ and in the $m$-regular
tree.
\end{proof}

\subsection{Achievable rates}

Let $\HH = \left(
  \begin{smallmatrix}
    \HH_X\\ \HH_Z
  \end{smallmatrix}
  \right)$ 
be the stabilizer matrix of a $(2,m)$ CSS code.
Remark that the graphs associated to $\HH_X$ and $\HH_Z$ have girth at
most $m$, since a row of $\HH_Z$ yields a cycle for $\HH_X$ and vice
versa.
We will say that a $(2,m)$ CSS code is proper if the two
associated graphs are connected and have girth exactly $m$.

We now translate the upper bound of
Proposition~\ref{prop:bound_fn_LDPC2} into a lower bound on the
function $D(p)$ of Theorem~\ref{theo:capacity_stabilizer} for a sequence of $(2,m)$ CSS codes.

\begin{prop} \label{prop:bound_g_LDPC2}
For any sequence of proper $(2,m)$ CSS codes we have:
$$
D(p) \geq \Big( \frac {1-2p} {1-p} \Big) 
\Big( \frac {4} {m} - 
\frac{4}{m} \big(1-(1-p)^m S_{m-2}( p(1-p)^{m-2} ) \big) \Big),
$$
where $S_{m-2}(z) = \sum_{k=0}^{m-2} \frac{a_k}{k+1} z^k$ and $a_k$ are the coefficients of the generating function $T_m$.
\end{prop}

\begin{proof}
  For $\HH$ the stabilizer matrix of a proper CSS code, we
  apply \eqref{eq:fconcave} by using for the upper bound $M(p)$ on
  $\Esp(\rank(\HH_\E))/n$, the sum of 
  the upper bounds provided by 
  Proposition~\ref{prop:bound_fn_LDPC2} on each of the terms
  $\Esp(\rank(\HH_{X,\E}))/n$ and $\Esp(\rank(\HH_{Z,\E}))/n$.
  The result follows after some rearranging.
\end{proof}

Finally, Proposition~\ref{prop:bound_g_LDPC2} together with
Theorem~\ref{theo:capacity_stabilizer} yields:

\begin{theo}\label{theo:capacity_CSS_LDPC2}
Over the quantum erasure channel of erasure probability $p$,
achievable rates of proper $(2, m)$ CSS codes satisfy 
$$
R \leq (1 - 2p) \Big( \frac{4}{mp} \big( 1 - (1-p)^{m} S_{m-2}(p(1-p)^{m-2}) \big) -1 \Big).
$$
where $S_{m-2}(z) = \sum_{k=0}^{m-2} \frac{a_k}{k+1} z^k$ and $a_k$ are the coefficients of the generating function $T_m$.
\end{theo}

%\begin{figure}[h]
%\centering
%\includegraphics[scale=0.5]{CSS_2_8}
%\caption{Upper bounds on achievable rates of CSS $(2,8)$ codes with Theorem \ref{theo:stabilizer} in blue and with Theorem \ref{theo:capacity_CSS_LDPC2} in black and the capacity of the quantum erasure channel in red.}
%\label{fig:CSS(2,8)}
%\end{figure}

\section{Erasure threshold of quantum LDPC codes} \label{section:threshold}

In this part, we reformulate our bound on achievable rates of quantum
LDPC codes by determining an upper bound on the erasure decoding threshold of regular quantum LDPC codes.

The capacity of the quantum erasure channel is $1-2p$. This means that
the rate of a family of quantum codes with vanishing decoding error probability over the quantum erasure channel of probability $p$, satisfies $R \leq 1-2p$.
Alternatively,
given a family of quantum codes of rate at least $R$, we can ask for
the highest erasure rate that we can tolerate with vanishing error
probability after decoding. This is the erasure decoding threshold.
Assume that we have a family of quantum codes of rate higher than $R$, which achieve an asymptotic zero error probability over the quantum erasure channel of erasure probability $p$. Then, we have:
$$
p \leq \frac{1-R}{2}.
$$
If we consider CSS codes of type $(2, m)$, we have $2n/m$ rows in each
matrix $\HH_X$ and $\HH_Z$. Therefore, the number of encoded qubits is
at least $(1-4/m)n$ (actually it is exactly $(1-4/m)n + 2/n$ by Lemma \ref{lemma:dimCX}).
Using the bound of Theorem~\ref{theo:capacity_CSS_LDPC2} and the fact that the rates of the quantum codes are higher than $1-4/m$, we obtain:
\begin{theo} \label{theo:erasure_threshold}
The erasure threshold of $(2, m)$ proper CSS codes is below the solution of the equation:
$$
1-\frac{4}{m} = (1 - 2p) \Big( \frac{4}{mp} \big( 1 - (1-p)^{m} S_{m-2}(p(1-p)^{(\ell-1)(m-1)-1}) \big) - 1 \Big)
$$
where $p \in [0, 1/2]$.
\end{theo}
This value is obtained at the intersection of the graphical
representation of the upper bound with the line $y = 1-4/m$. An
example of these curves is given in Figure~\ref{fig:CSS_2_8_v2}.

\begin{figure}[h]
\centering
\includegraphics[scale=0.5]{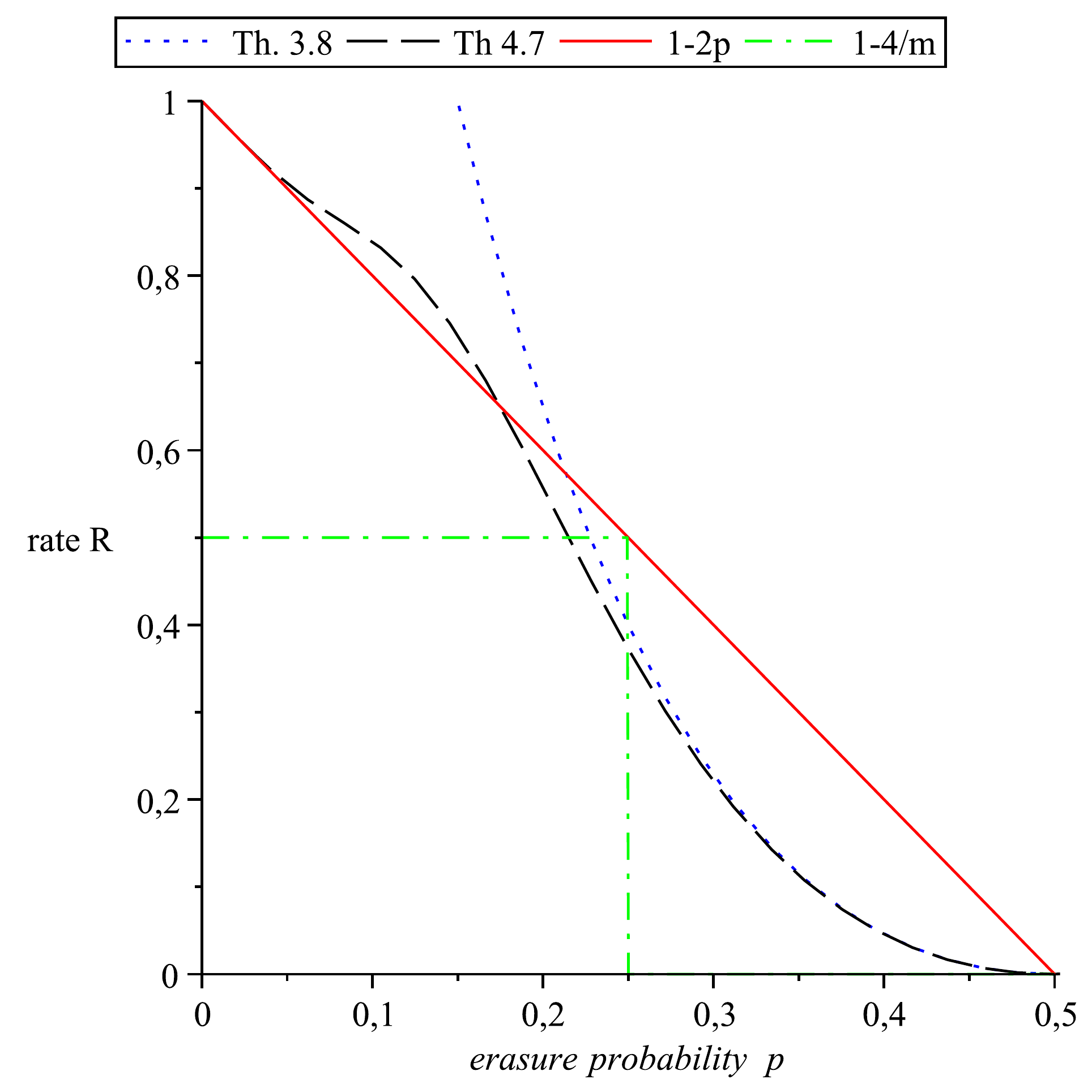}
\caption{The horizontal green line is the rate of a proper CSS (2,8)
  code. Its intersection with the capacity gives an upper bound on the
  erasure threshold of quantum codes : $p_e \leq 0.25$. The two other
  curves are the upper bounds of Theorems~\ref{theo:stabilizer} and 
  \ref{theo:erasure_threshold}: their intersection with the horizontal
  line gives the two upper bounds $p_e \leq 0.228$ and $p_e \leq 0.215$ for CSS (2,8) codes from Theorem \ref{theo:erasure_threshold}.}
\label{fig:CSS_2_8_v2}
\end{figure}

\bigskip

Using symbolic computation software, we computed different numerical
values of this upper bound on the decoding erasure threshold $p_e$:
$$
\begin{array}{|c|c|c|}
\hline
\text{ type } & \text{improved upper bound on } p_e & \text{ Capacity bound } p_e \leq 2/m\\
\hline
CSS(2,8), \text{ with Th. } \ref{theo:stabilizer} & 0.228 & 0.25 \\
CSS(2,8), \text{ with Th. } \ref{theo:erasure_threshold} & 0.215 & 0.25\\
Stab(4,8), \text{ with Th. } \ref{theo:stabilizer} & 0.228 & 0.25\\
CSS(2,5), \text{ with Th. } \ref{theo:stabilizer} & 0.387 & 0.40 \\
CSS(2,5), \text{ with Th. } \ref{theo:erasure_threshold} & 0.381 & 0.40 \\
\hline
\end{array}
$$

% CSS codes of type $(2,8)$ are also Stabilizer codes of type $(4, 8)$. These families have rate at least 1/2 and erasure threshold at most 0.25. We obtain a more accurate result for CSS codes.
% We can also remark that for $\ell \leq 3$ the most important contribution is given by components of size 0. The other components will be negligible in the bound.
% We can obtain a similar result for stabilizer regular LDPC codes using Theorem \ref{theo:capacity_stabilizer_LDPC}

\section{Application to percolation theory} \label{section:percolation}

% In the paper \cite{DZ10}, we obtained a bound on the critical probability of hyperbolic graphs $G_m$ using the upper bound on the capacity of the quantum erasure channel. This result was deduced from the construction of surfaces codes with finite quotient of the graphs $G_m$. This part is devoted to the improvement of this bound on the  critical probability using the bound of Theorem~\ref{theo:capacity_CSS_LDPC2}.

\subsection{Percolation theory}

In this section $E$ denotes the edge set of a graph, rather than a
Pauli error: context should not allow confusion.
Let $G = (V, E)$ be an infinite graph. Denote by $\mu_p$
the probability measure on $\{0,1\}$ defined by $\mu_p(\{1\})=p$.
Consider the product space $\Omega = \{0,1\}^E$ endowed with the
product probability measure $\Prob_p=\mu_p^{\otimes E}$.
Random events should be seen as subgraphs. Informally, we choose
every edge of $G$ with probability $p$ independently of the other
edges, and obtain a random subgraph. The edges of this subgraph are
called {\em open} edges.
Percolation theory is interested in the probability 
that a given edge $e$ is contained in a infinite open connected component
(an open {\em cluster}). This probability depends a priori on the edge
$e$, but not if the graph $G$ is edge-transitive, for example if
$G$ is the infinite square lattice (Figure~\ref{fig:square}).
The central parameter in percolation theory is the {\em critical probability} $p_c$, defined as:
$$
p_c(G) = \inf\{p \in [0, 1],\, \Prob_p(|\E(e)| = \infty)>0 \},
$$
where $\E(e)$ denotes the open cluster containing edge $e$.

\begin{figure}[h]
  \centering
  \begin{tikzpicture}[scale=0.9]
  
\draw[dashed]
\foreach \x in {1,2,...,5} {  (\x,0.25)-- (\x,5.75) };
\draw[dashed]
\foreach \y in {1,2,...,5} {  (0.25,\y)-- (5.75,\y) };
\draw[step=1cm] (1,1) grid (5,5);

  \end{tikzpicture}
  \caption{The square lattice}
  \label{fig:square}
\end{figure}
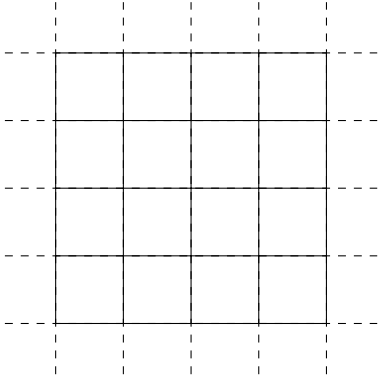

By a famous result of Kesten \cite{K80} that stayed a conjecture for 20 years,
we have $p_c=1/2$ for the square lattice. Computing the critical
probability exactly is usually quite difficult.
%but one class of
%graphs for which percolation is fairly
%well understood is trees: in particular it is straightforward to
%compute the critical probability of a regular tree of degree $\Delta$,
%in which case we have $p_c=1/(\Delta -1)$.

For any integer $m \geq 4$, we denote
by $G(m)$ the planar graph which is regular of degree $m$ and
tiles the plane by elementary faces of length~$m$. For $m=4$ the graph
$G(4)$ is exactly the square lattice. The local structure of the graph
$G(5)$ is shown on Figure~\ref{fig:G(5)}.

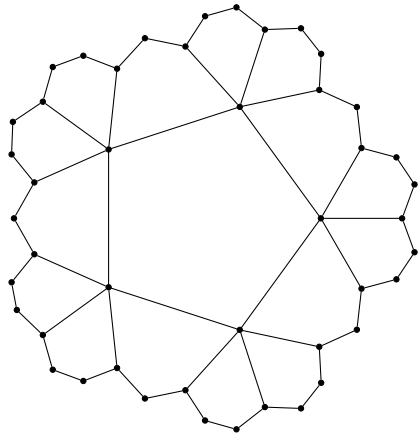
\begin{figure}[h]
  \centering
  \begin{tikzpicture}[scale=1.2]

\tikzstyle{every node}=[circle, draw, fill=black!100, inner sep=0pt, minimum width=2pt]

%pentagone central
\draw
(0:1.3) \foreach \x in {72,144,...,359} { -- (\x:1.3) } -- cycle
(0:1.3) node (a) {}
(72:1.3) node (b) {}
(144:1.3) node (c) {}
(216:1.3) node (d) {}
(288:1.3) node (e) {};

%on ajoutes les aretes autour des sommets du pentagone
\draw[shift=(a)]
(0:.9) node (a2) {}
(60:.9) node (a3) {}
(300:.9) node (a1) {};
\draw
(a1)--(a)--(a2)
(a)--(a3);

\draw[shift=(b)]
(72:.9) node (b2) {}
(132:.9) node (b3) {}
(12:.9) node (b1) {};
\draw
(b1)--(b)--(b2)
(b)--(b3);

\draw[shift=(c)]
(144:.9) node (c2) {}
(204:.9) node (c3) {}
(84:.9) node (c1) {};
\draw
(c1)--(c)--(c2)
(c)--(c3);

\draw[shift=(d)]
(216:.9) node (d2) {}
(276:.9) node (d3) {}
(156:.9) node (d1) {};
\draw
(d1)--(d)--(d2)
(d)--(d3);

\draw[shift=(e)]
(288:.9) node (e2) {}
(348:.9) node (e3) {}
(228:.9) node (e1) {};
\draw
(e1)--(e)--(e2)
(e)--(e3);

%fermeture des faces qui contiennent une arete du pentagone
\draw
(36:2.1) node (f){}
(108:2.1) node (g){}
(180:2.1) node (h){}
(252:2.1) node (i){}
(324:2.1) node (j){};
\draw
(a3)--(f)--(b1)
(b3)--(g)--(c1)
(c3)--(h)--(d1)
(d3)--(i)--(e1)
(e3)--(j)--(a1);

%fermeture des faces qui on un sommet en commun avec le pentagone
\draw[shift=(a2)]
(70:.4) node (a24){}
(290:.4) node (a21){};
\draw[shift=(a3)]
(345:.4) node(a31){};
\draw[shift=(a1)]
(15:.4) node (a14){};
\draw
(a2)--(a24)--(a31)--(a3)
(a1)--(a14)--(a21)--(a2);

\draw[shift=(b2)]
(142:.4) node (b24){}
(2:.4) node (b21){};
\draw[shift=(b3)]
(57:.4) node(b31){};
\draw[shift=(b1)]
(87:.4) node (b14){};
\draw
(b2)--(b24)--(b31)--(b3)
(b1)--(b14)--(b21)--(b2);

\draw[shift=(c2)]
(214:.4) node (c24){}
(74:.4) node (c21){};
\draw[shift=(c3)]
(129:.4) node(c31){};
\draw[shift=(c1)]
(159:.4) node (c14){};
\draw
(c2)--(c24)--(c31)--(c3)
(c1)--(c14)--(c21)--(c2);

\draw[shift=(d2)]
(286:.4) node (d24){}
(136:.4) node (d21){};
\draw[shift=(d3)]
(201:.4) node(d31){};
\draw[shift=(d1)]
(231:.4) node (d14){};
\draw
(d2)--(d24)--(d31)--(d3)
(d1)--(d14)--(d21)--(d2);

\draw[shift=(e2)]
(358:.4) node (e24){}
(218:.4) node (e21){};
\draw[shift=(e3)]
(273:.4) node(e31){};
\draw[shift=(e1)]
(303:.4) node (e14){};
\draw
(e2)--(e24)--(e31)--(e3)
(e1)--(e14)--(e21)--(e2);

  \end{tikzpicture}
  \caption{The local structure of the graph $G(5)$}
  \label{fig:G(5)}
\end{figure}

For $m>4$ these graphs make up regular tilings of the hyperbolic
plane. Interest in percolation on hyperbolic tilings was raised in a number
of papers e.g. \cite{Be96, BMK09, GZ11} and determining their critical probability $p_c(m)$ is 
highly non-trivial. Note that all graphs $G(m)$ are self-dual like
the square lattice $G(m)$.

We have the following easy bounds on $p_c$:

\begin{prop} \label{prop:easy_bound}
The critical probability $p_c$ of $G(m)$ satisfies
  $$\frac{1}{m-1} \leq p_c.$$
\end{prop}

\begin{proof}
  We adapt the proof of \cite{Gr} page 14 in the case of the square lattice.

Let O be a fixed vertex. To show the first inequality we can say that
there are not more than $m(m-1)^{n-1}$ paths from O of length $n$ in
$G(m)$ and the probability of such an open path is $p^n$. So if $p <
\frac{1}{m-1}$ 
the average length of an open path from O is not more than $\sum_{n=1}^\infty m(m-1)^{n-1}p^n < \infty$. In this case $p$ is under the critical probability.
\end{proof}

The same method leads to the upper bound $p_c  \leq 1-\frac{1}{m-1}$.
The proof can be immediatly adapte from the case of the square lattice \cite{Gr} page 14.

\subsection{Quotient graphs}

To study percolation on the hyperbolic tiling $G(m)$, we need a family
of increasingly big finite graphs which are locally the same as $G(m)$. We will use a family introduced by \v{S}ir\'a\v{n} in~\cite{Si00}.

Let $P_k(X) = 2\cos(k\arccos(X/2))$ be the $k$-th normalized Chebychev polynomial and $\xi = 2\cos(\pi/m^2)$. Let $y$ and $z$ be the matrices of $SL_3(\Z[\xi])$ defined by
$$
y  =
\left(
\begin{array}{ccc}
P_m(\xi)^2-1 & 0 & P_m(\xi)\\
P_m(\xi) & 1 & 0\\
-P_m(\xi) & 0 & -1
\end{array}
\right)
$$

$$
z =
\left(
\begin{array}{ccc}
-1 & -P_m(\xi) & 0\\
P_m(\xi) & P_m(\xi)^2-1 & 0\\
P_m(\xi) & P_m(\xi)^2 & 1
\end{array}
\right).
$$
These two matrices generate the triangular group $T(m)$ \cite{Si00}. 
To obtain a finite graph we can reduce the entries of the matrices
modulo a prime number $p$. The coefficients are in the ring $\Z[\xi]$
which is isomorphic to the quotient $\Z[X]/h(X)$ where $h$ is the
minimal polynomial of the algebraic integer $\xi$. 
Reducing coefficients modulo $p$, we obtain a group homomorphism from
$SL_3(\Z[\xi])$ to $SL_3(\F_p[X]/(h(X))$. The image of $T(m)$
 will be called $\bar T(m)$. 

Let $\bar G(m)$ be the graph defined like $G(m)$ but with the group $\bar
T(m)$, in other words the vertices, edges and faces of $\bar G(m)$ are
defined as the left cosets of $\gen{\bar y}$, $\gen{\bar y \bar z}$ and
$\gen{\bar z}$ respectively.
There is a surjection $s$ from $G(m)$ to $\bar G(m)$ which sends
$u\gen{y}$ to 
$\bar u \gen{\bar y}$. 

Following \v{S}ir\'a\v{n}, let us define the {\em injectivity radius}
of the graph $\bar G(m)$ as the largest 
integer $r$ such that the restriction of the surjection $s$ to a ball
of radius $r$ is one-to-one. It is shown in \cite{Si00} that
we can choose $p$ so as to have $r$ arbitrarily large. Loosely
speaking, \v{S}ir\'a\v{n}'s argument is that 
if two distinct vertices $u\gen{y}$ and $v\gen{y}$ 
in $G(m)$  have the same image under $s$ then $u^{-1}v$ in $T(m)$ must
project to the identity element in $\bar T(m)$. But this means that
the matrix $u^{-1}v$ has polynomial entries that, properly reduced
modulo $h(X)$, can only be expressed
with coefficients at least one of which exceeds $p$: this implies that $u^{-1}v$ can
only be expressed as a product of a large number of matrices $y$ and
$z$, which in turn means that the original vertices $u\gen{y}$ and $v\gen{y}$ 
have to be far apart in $G(m)$.

The above construction enables us
to define a family of finite graphs $(G_r(m))_{r\geq 1}$ such that each
graph $G_r(m)$ has injectivity radius at least $r$, 
for every integer $r$. 

Let us now define random subgraphs of $G_r(m)$
through the product measure $\mu_p^{\otimes E_r}$, where $E_r$ denotes
the edge set of $G_r(m)$. In other words the open subgraph of
$G_r(m)$ is created by declaring every edge open with 
independent probability $p$.

For any fixed edge $e$, let $\E_r(e)$ be the (possibly empty)
connected component of the random subgraph of $G_r(m)$ that contains
$e$ and call it again the open cluster containing $e$.
Let $f_r(p)$ be the probability that $|\E_r(e)|>r$. We have:

\begin{prop}\label{prop:upper_bound_f}
If $p < p_c(m)$ then $f_r(p)$ goes to $0$ when $r$ goes to infinity.
\end{prop}

\begin{proof}
Notice that the probability $1-f_r(p)$ that the open cluster
containing $e$ has cardinality not more than $r$ is the same
for the random subgraph defined on the finite graph $G_r(m)$ 
and the random subgraph defined on the infinite graph $G(m)$. This is
because this event depends only on the ball of radius $r$ centered on an
endpoint of $e$, and these balls in $G_r(m)$ and $G(m)$ are
isomorphic.

We can therefore consider $f_r(p)$ to mean the probability of the
event $F_r$ that $|\E(e)|>r$
in the infinite graph $G(m)$. Now $(F_r)_{r\geq 1}$ is a
decreasing sequence of events, and $\Prob_p(\cap_{r\geq 1} F_r)$ is
exactly the probability of percolation, which is $0$ since we have
supposed $p<p_c$. 
By monotone convergence we therefore have
$f_r(p)=\Prob_p(F_r)\rightarrow 0$.
\end{proof}

\subsection{The quantum codes $Q_r(m)$ associated with the graphs $G_r(m)$}

Every finite graph $G_r(m)$ gives rise to a CSS quantum code 
$Q_r(m)$ whose
coordinate set is the edge set $E$ of the graph. We will have
therefore a quantum code of length $n=|E|$. The matrices $\HH_X$ and
$\HH_Z$ are defined as described in Section~\ref{section:2complex}.
The rows of $\HH_X$ are in one-to-one
correspondence with the vertices of the graph. Every vertex $x$
yields a row of $\HH_X$ whose support is exactly the set of edges
incident to $x$. Every row of $\HH_X$ therefore has weight $m$.
The rows of the other matrix $\HH_Z$ is in one-to-one correspondence
with the set of faces of the graph. Every face yields a row whose
support is equal to the set of edges making up the face. Since faces
are $m$-gons, every row of $\HH_Z$ also has weight $m$.
It should be clear that rows of $\HH_X$ and $\HH_Z$ meet in either
$0$ or $2$ edges, so any row of $\HH_X$ is orthogonal to any row of
$\HH_Z$ and we have a quantum CSS code of type $(2,m)$.

Recall from Section~\ref{section:2complex} that
the classical code $C_X$ is the cycle code of the graph $G_r(m)$.
When we reverse the roles of $\HH_X$ and $\HH_Z$ by declaring the rows
of $\HH_Z$ (rather than those of $\HH_X$) to be vertices and the rows
of $\HH_X$ to be faces, the graph thus defined is called the dual
graph of $G_r(m)$ and we denote it here by $G_r^*(m)$. The classical code
$C_Z$ is thus the cycle code of the dual graph $G_r^*(m)$.

From Lemma~\ref{lemma:dimCX}, since the graphs $G_r(m)$ and
$G_r^*(m)$ are connected, we have:

\begin{prop}\label{prop:dim}
The dimension $k$ of the quantum code $Q_r(m)$ equals:
$$k=\left(1-\frac{4}{m}\right)n+2.$$
\end{prop}

We remark that for $m=4$, the graph $G_r(4)$ is a combinatorial torus and
the quantum code $Q_r(4)$ is a version of Kitaev's toric code
\cite{Ki97}. For $m\geq 5$
the quantum codes $Q_r(m)$
have positive rate bounded away from zero and minimum distance at
least $2r$ (see the remark after the proof of
Proposition~\ref{lemma:upper_bound_g} below)
which is a quantity which behaves as $\log n$. See \cite{Ze09} for a
discussion of similar families of surface codes.

\medskip

\noindent
{\bf Non-correctable erasures.}
An erasure vector can be identified with a set of
edges of $G_r(m)$ (or of $G_r^*(m)$) and we will denote it as before
by $\E$.
From Proposition~\ref{prop:noncorrectable} we have that the erasure
pattern $\E$ is non-correctable if and only if $\E$ either contains a cycle of $G_r(m)$ which is not
a sum of faces (an element of $C_X\setminus C_Z^\perp$) or $\E$, 
viewed as a set of edges of the dual graph
$G_r^*(m)$, contains a cycle of $G_r^*(m)$ that is not a
sum of faces of $G_r^*(m)$ (an element of $C_Z\setminus C_X^\perp$).

\subsection{Bounds on the critical probability using the capacity of the quantum erasure channel}
Consider an arbitrary member of
the family of quantum codes $Q_r(m)$ associated with the graphs
$G_r(m)$. 
Because the original graph $G(m)$ is
self-dual, all arguments involving $G_r(m)$ will be seen to hold
for its dual graph $G_r^*(m)$ and we will focus on the probability
that the random erasure pattern $\E$ contains a cycle that is not
a sum of faces in the original graph~$G_r(m)$.

We would like to derive the upper bound on $p_c$ in
Theorem~\ref{theo:percolation_bound} below by claiming the following: if $p<p_c$, then
for the family of graphs $G_r(m)$, the probability that the random
set of edges~$\E$ contains a cycle which is not a sum of faces
vanishes. If this is true, then Theorem~\ref{theo:capacity_stabilizer}
applies and the rate $R$ of the quantum code $Q_r(m)$
must satisfy $R<1-2p-D(p)$ for every $p<p_c$ and
Proposition~\ref{prop:dim} gives the result since $R=1-4/m$.

Unfortunately, we do not know whether for every $p<p_c$, 
the erasure pattern $\E$ contains no cycle that is not a sum of faces
with high probability. What we will prove however, is that if $\E$
contains a cycle that is not a sum of faces, then with high
probability one of the
representatives of this cycle modulo the space of faces must
have very small weight. To violate the capacity of the erasure channel
we will therefore use, not $Q_r(m)$ directly, but an ``improved''
version $Q_r'(m)$ of $Q_r(m)$ that we now introduce.

\begin{prop}\label{prop:code_improvement}
Let $Q_r(m)$ be a hyperbolic code, $n$ its length and $R$ its rate.
Suppose $\rho \in ]0, \frac{1}{2}[$ and $\alpha \in ]0, 1[$ are such that
$$
h(\rho) < \alpha < \frac{R}{2},
$$
where $h(\rho)=-\rho\log_2\rho -(1-\rho)\log_2(1-\rho)$ denotes the binary entropy function.
Then we can add $\alpha n$ rows to the parity-check matrix $\HH_X$ and
$\alpha n$ rows to the parity-check matrix $\HH_Z$ 
of $Q_r(m)$ to obtain a CSS code $Q_r'(m)$ of length $n$, rate
$R-2\alpha$ and distance $d \geq \rho n$. 
\end{prop}

\begin{proof}
Denote by $r_X$ and $r_Z$ the dimension of the code $C_X^\perp$ and $C_Z^\perp$ respectively. We have $r_X = r_Z = \frac{2}{m} n -1$.

We will construct a matrix $\HH_X'$ by adding $\alpha n$ rows to the
matrix $\HH_X$ such that the rows of $\HH_X'$ are orthogonal to the
rows of $\HH_Z$ and the rank of $\HH_X'$ is $r_X+\alpha n$. Let $C_X'$
be the code of parity-check
matrix $\HH_X'$.

For $\rho \in ]0, 1/2[$, we define $X_\rho$ by
$$
X_\rho(\HH_X') = |\{ v \in  C_X' \backslash C_Z^\perp | w(v) \leq \rho n \}|.
$$

We can write $X_\rho$ as a sum a random variables to see that
$$
\mathbb E (X_\rho) = \sum_{\substack{v \in C_X\backslash C_Z^\perp \\ v \in B(0, \rho n)}} \frac{|\{\HH_X' | v \in C_X'\}|}{|\{\HH_X'\}|}.
$$
where $B(0, \rho n)$ denotes the Hamming ball of radius $\rho n$.
Let $L_1, L_2, \dots L_{r_X}$ be $r_X$ rows of $\HH_X$.
The number of suitable matrices $\HH_X'$ is the number of families $L_1', L_2', \dots L_{\alpha n}'$ of vectors of $\F_2^n$ such that $L_j' \in C_Z$ for all $j$ and $(L_1, L_2, \dots L_{r_X}, L_1', L_2' \dots, L_{\alpha n}')$ are linearly independant.

We can construct a suitable matrix $\HH_X'$ if and only if $r_X + \alpha n \leq \dim(C_Z)$ this gives the condition $\alpha < (1-\frac{4}{m})-\frac{2}{n}$.
In this case the number of matrices is
$$
\prod_{i=r_X}^{r_X+\alpha n-1} (2^{n-r_Z}-2^i).
$$

To evaluate the cardinality $|\{\HH_X' | v \in C_X'\}|$ with $v$ in $C_X \backslash C_Z^\perp$, it suffices to add the condition $L_j' \in \{v\}^\perp$ for all $j$.
We get
$$
\prod_{i=r_X}^{r_X+\alpha n-1} (2^{n-r_Z-1}-2^i).
$$

So we have
$$
\frac{|\{\HH_X' | v \in C_X'\}|}{|\{\HH_X'\}|} = \frac{2^{n-r_X-r_Z-\alpha n}-1}{2^{n-r_X-r_Z}-1} \leq 2^{-\alpha n}.
$$
This bound doesn't depend of $v$ so we can give an upper bound on the expectation of $X_\rho$ because we know that the number of words in the ball of radius $\rho n$ is less than $2^{nh(\rho)}$.
We find
$$
\mathbb E (X_\rho) \leq 2^{n(h(\rho)-\alpha)}.
$$
If $\alpha > h(\rho)$ the mean goes to 0.
Since $X_\rho$ has integer values there exists $\HH_X'$ such that $X_\rho(\HH_X) = 0$.
We obtain a CSS code of matrix $\HH_X'$ with $r_X' = r_X+\alpha n$ and $\HH_Z$ unchanged such that the minimum weight of a word of $C_X' \backslash C_Z^\perp$ is at least $\rho n$.

We want to repeat this argument to have the minimum weight of a word of $C_Z' \backslash C_X'^\perp$ higher than $\rho n$. It suffices to choose $\alpha < \frac{1}{2}(1-\frac{4}{m})+\frac{1}{n}$ because in this case $r_Z + \alpha n < \dim(C_X)$.
\end{proof}

Let $\E$ be an erasure. We can write 
\begin{equation}
\label{eq:EC}
\E = \E_C + \E_P
\end{equation}
where $\E_C$ is the sum of the connected components which do not cover
a cycle which is not a sum of faces. The {\em problematic} part $\E_P$
of $\E$ is the union of the others components.

In the graph $G_r(m)$, define $g_r(p)$ to be the probability that
that the open cluster $\E_r(e)$ covers a cycle which is not a sum of
faces.
We have:
\begin{lemma}\label{lemma:upper_bound_g}
If $p<p_c(m)$ then $g_r(p)$ goes to $0$ when $r$ goes to infinity.
\end{lemma}

\begin{proof}
Recall that $f_r(p)$ denotes the probability that $|\E_r(e)|>r$. 
We prove that $g_r(p) \leq f_r(p)$ and apply Proposition~\ref{prop:upper_bound_f}.
If $|\E_r(e)|\leq r$ then the open cluster 
$\E_r(e)$ is included in a ball of radius $r$ of the graph $G_r(m)$. 
Since this ball is isomorphic to the ball of the same radius 
in the planar graph $G(m)$, it is planar. In any planar graph every
cycle is a sum of faces so $\E_r(e)$ covers a cycle which is not a sum
of faces only if $|\E_r(e)|> r$, hence $g_r(p) \leq f_r(p)$.
\end{proof}

\noindent
{\bf Remark:}
By the same planarity argument as above,
every cycle of length less than $2r$ in the graph $G_r(m)$ is a sum of faces. This proves that the distance of the quantum code $Q_r(m)$ is at least $2r$.

\begin{prop}\label{prop:bound_mean}
If we consider the erasure channel of probability $p<p_c$ then
$\forall \e>0, \exists r_0 \in \N$ such that if $r \geq r_0$ then the 
expectation of the weight of $\E_P$ defined as in \eqref{eq:EC} satisfies
$$
\mathbb E (|\E_P|) \leq \e n.
$$
\end{prop}

\begin{proof}
For any edge $e$ of $G_r(m)$, let $X_{r, e}$ be the random variable
which take the value 1 if the connected component $\E_r(e)$ of $e$ in
$G_r(m)$ covers a cycle which is not a sum of faces and the value 0
otherwise. Then we have:
$$|\E_P| = \sum_eX_{r, e}.$$
To conclude note that $\mathbb E (X_{r, e})=g_r(p)$ and apply
Lemma \ref{lemma:upper_bound_g}.
\end{proof}

The next Lemma states that if the erasure vector $\E$ has a large
``problematic'' part $\E_P$ then it must be correctable by the 
``improved'' codes given by Proposition~\ref{prop:code_improvement}.
\begin{lemma}
  Let $Q'_r(m)$ be one of the quantum codes given by Proposition~\ref{prop:code_improvement} and let $d$ be its minimum distance. Suppose the part $\E_P$ of the
erasure vector $\E$ defined in \eqref{eq:EC} satisfies $|\E_P|<d$. Then
$\E$ is correctable by $Q'(m)$.
\end{lemma}

\begin{proof}
Denote by $C_X$ and $C_Z$ the binary linear codes associated with the 
quantum code $Q_r(m)$ 
and by $C_X', C_Z'$ their  binary sub-codes associated with the quantum
code $Q_r'(m)$ introduced in Proposition~\ref{prop:code_improvement}
and defined by augmenting the 
parity-check matrices $\HH_X$ and $\HH_Z$ of $Q_r(m)$.

If the erasure vector
$\E$ covers an element $x$ of $C_X'\backslash C_Z'^\perp$ then $x$
must belong to $C_X \backslash C_Z^{\perp}$ \emph{i.e.} $x$
is a cycle of $G_r(m)$ which is not a sum of faces. 
The restriction of this cycle to $\E_C$ defined
in \eqref{eq:EC}
is another cycle $y$ and the definition of $\E_C$ implies that $y$ is
a sum of faces. We obtain that $x+y$ is included in $\E_P$ with
$y\in C_Z^{\perp}\subset C_Z'^\perp$, i.e. $x+y\in C_X'\backslash
C_Z'^\perp$ but this is a contradiction whenever 
the part $\E_P$ of the erasure $\E$ has weight strictly less than the 
minimum distance $d$ of the improved code $Q'(m)$.
\end{proof}

We are now in a position to give an upper bound on the critical
probability of $G(m)$. Recall the definition of the rank difference
function of a family of stabilizer codes (Definition~\ref{defi:D}).

\begin{theo} \label{theo:percolation_bound}
Let $m\geq 5$, and let $p_c(m)$ be the critical probability for
percolation on $G(m)$. Let $D(p)$ be the rank difference function
of the sequence of stabilizer matrices associated to the tilings $G_r(m)$.
Then for any $p<p_c$ we have:
$$
1-\frac 4m \leq 1-2p-D(p).
$$
\end{theo}

\begin{proof}
Let $R = 1-\frac{4}{m}$ and fix $p<p_c$.
For any $\alpha$ such that $0 < \alpha <R/2$,
Proposition~\ref{prop:code_improvement} gives us a quantum code
$Q'(m)$ with minimum distance $d\geq \rho n$ where
$\rho=h^{-1}(\alpha/2)$ and rate $R-2\alpha$. For such a code
the probability of a decoding error satisfies:
$$P_{err}  \leq  P(|\E_P| \geq \rho n).$$

For any $\e >0$ we can take $r$ large enough so that Proposition
\ref{prop:bound_mean} applies, and together with  Markov's inequality we have
$$
P_{err} \leq P(|\E_P| \geq \frac{\rho}{\e} \e n) \leq \frac{\e}{\rho}.
$$
For every $\e>0$ we take $\rho = \sqrt{\e}$. Then $\rho(\e)$ and
$\frac{\e}{\rho(\e)}$ simultaneously go to zero when $\e$ goes to
zero. Defining $\alpha$ by $\alpha = 2h(\rho)$ and choosing a
decreasing sequence
of $\epsilon$'s that tends to zero, we obtain a family of quantum
codes $Q_r'(m)$ with decoding error probability tending to zero and
rate $R-2\alpha$ tending to $R$.

We can therefore apply Theorem~\ref{theo:capacity_stabilizer} to the
sequence of stabilizer matrices of the codes $Q_r'(m)$. But we have
just seen that their rates tend to $1-4/m$ and furthermore, since the codes $ Q'_r(m)$ are obtained from the codes
$Q_r(m)$ by adding a vanishing proportion of generators to their
stabilizer group, the function $D(p)$ is the same for the family
$Q'_r(m)$ as for the family of surface codes $Q_r(m)$.
\end{proof}

Applying the lower
bound on $D(p)$ stemming from Proposition~\ref{prop:bound_g_LDPC2}, we
obtain after some rearranging:

\begin{theo}\label{theo:percolation_bound_LDPC}
We have $p_c(m) \leq p$ where $p$ is the smallest solution $p \in [0, 1]$ of the equation:
$$
1 - \frac{4}{m} = (1 - 2p) \Big( \frac{4}{mp} \big( 1 - (1-p)^{m} S_{m-2} (p(1-p)^{m-2}) \big) -1 \Big).
$$where $S_{m-2}(x) = \sum_{k=0}^{m-2} \frac{a_k}{k+1} x^k$ and $a_k$ are the coefficients of the generating function $T_m$.
\end{theo}

Using symbolic computation software, we can compute this bound on the critical probability.
We use classical properties of generating functions to compute elements of the sequence $(a_k)_k$. For classical theorems on generating functions, in particular Lagrange inversion theorem, see for example \cite{FS09}, \cite{Wi06}.

$$
\begin{array}{|c|c|c|c|}
\hline
m & \text{ lower lower on } p_c(m): \frac{1}{m-1} & \text{bound of Th. \ref{theo:percolation_bound_LDPC}} & \frac{2}{m} = \text{bound of Prop. \ref{prop:easy_bound} }\\
\hline
5 & 0.25 & 0.38 & 0.40\\
10 & 0.11 & 0.16 & 0.20\\
20 & 0.053 & 0.073 & 0.100\\
30 & 0.035 & 0.046 & 0.067\\
40 & 0.026 & 0.033 & 0.050\\
50 & 0.020 & 0.026 & 0.040\\
\hline
\end{array}
$$

We can observe that the new upper bound becomes better for tilings $m$ with faces of large length. It is not surprising because, in this case, we enumerate the connected components of larger size.

The exact value of the critical probability remains to discover.
Numerical estimations of this value are difficult due to the
exponential growth of balls of the graph. 
For example, Ziff pointed out the inconsistency of some numerical results in \cite{GZ11}.
This reinforces the importance of our theoretical approach.

\section{Concluding Remarks}

\begin{itemize}
\item We have given a combinatorial proof of the upper bound on
  achievable rates of stabilizer codes $R\leq 1-2p$, over the quantum
  erasure channel. This proof is of course less general than previous
  proofs, since it applies only to stabilizer codes, but it gives
  mathematical insight into the degeneracy phenomenon and, as we have
  shown by the study of sparse stabilizer codes, it has the potential
  for improved upper bounds on achievable rates for particular classes
  of quantum codes. A generalization of this approach to the
  depolarizing channel would be very welcome. The difficulty is the fact that for depolarizing noise the probability of an error depends on its weight. We must deal with typical errors.

\item By graphical arguments, we proved that stabilizer and CSS codes defined by generators of bounded weight don't achieve the capacity of the quantum erasure channel. This result can be applied to surfaces codes, color codes and a lot of well studied families of quantum codes. This encourages us to construct irregular quantum LDPC codes and families with growing weights.

\item The exact value of the critical probability of hyperbolic
  tilings remains unknown. To improve our upper bound, we must enumerate more
  connected components in the random graph, which becomes more
  difficult when we count components that are not subtrees.
  Our method also involves a concavity argument for the mean rank
  function which relates it to the rank difference function. This
  results in a manageable lowerbound on the relative rank difference function
  but it is generally not tight and alternative ways of evaluating
  this function would be desirable.
\end{itemize}

%the case l=2 for stabilizers
%classical case ?
%generalization to subsystem codes

\section*{Acknowledgement}

This work was supported by the French ANR Defis program under contract
ANR-08-EMER-003 (COCQ project). We acknowledge support from the Délégation Générale pour l'Armement (DGA) and from the Centre National de la Recherche Scientifique (CNRS).
We are very greatful to Jean-François Marckert for helpful discussion around generating functions.

\appendix
\makeatletter
\def\@seccntformat#1{Appendix~\csname the#1\endcsname:\quad}
\makeatother

\section{The rank of a random submatrix} \label{appendix:concavity_inequality}

The goal of this section is to prove that the function $\phi(p) =
\frac{1}{n} \mathbb E_p(\rank \HH_\E)$ is a concave ($\cap$-convex)
function. 
Then, we will use the concavity of $\phi$ to obtain a lower bound on $\Delta(p) = \phi(1-p) - \phi(p)$.

The key argument is the submodularity of the rank:
\begin{lemma} \label{lemma:submodularity}
Let $\HH \in M_{r, n}(\Pauli_1)$ be a Pauli matrix with $n$ columns. The rank function:
\begin{align*}
\mathcal P (\{1, 2, \dots, n\}) & \longrightarrow \N\\
A & \longmapsto \rank(A) = \rank(\HH_A)
\end{align*}
is a submodular function. That is, the rank function satisfies:
$$
\rank (A \cap B) + \rank (A \cup B) \leq \rank(A) + \rank(B).
$$
\end{lemma}

This lemma embraces the case of a binary matrix.
These rank properties can be obtained in the more general framework of matroid theory. A classical book on matroid theory is \cite{Ox92}.

\begin{proof}
We will prove the lemma for binary matrices. Then, we will explain how to adapt our argumentation to the quaternary case.

Let $A$ and $B$ be two subsets of $\{1,2, \dots, n\}$.
From the dimension formula for the sum of two subspaces, we deduce:
\begin{center}
$
\begin{cases}
\rank(A \cup B) = \rank(A) + \rank(B \backslash A) - \dim (\im H_A) \cap (\im H_{B \backslash A})\\
\rank(A \cup B) = \rank(B) + \rank(A \backslash B) - \dim (\im H_B) \cap (\im H_{A \backslash B})\\
\rank(A \cap B) = \rank(A) - \rank(A \backslash B) 
+ \dim (\im H_{A \cap B}) \cap (\im H_{A \backslash B})\\
\rank(A \cap B) = \rank(B) - \rank(B \backslash A) 
+ \dim (\im H_{A \cap B}) \cap (\im H_{B \backslash A})\\
\end{cases}
$
\end{center}
Look at the last term of these equalities.
We have clearly $\im H_{A \cap B} \subset \im H_A$, therefore the space $(\im H_{A \cap B}) \cap (\im H_{B \backslash A})$ is a subspace of $(\im H_A) \cap (\im H_{B \backslash A})$. This proves the dimension inequality:
$$
\dim (\im H_{A \cap B}) \cap (\im H_{B \backslash A}) \leq \dim (\im H_A) \cap (\im H_{B \backslash A}).
$$
We have the same result exchanging $A$ and $B$.
We sum the four equalities and we apply the above inequality. We get the desired result:
$$
\rank(A \cap B) + \rank(A \cup B) \leq \rank(A) + \rank(B).
$$

To prove the property for a matrix with coefficients in $\Pauli_1$, it is sufficient to show that all the tools from linear algebra used above are still satisfied for a Pauli matrix $\HH$.
As recalled in Section \ref{subsection:rank}, there is an isomorphism of $\F_2$ vector spaces between the Pauli group $\Pauli_n$ and the space $\F_2^{2n}$.
Therefore, we can regard the matrix $\HH \in \mathcal M_{r, n}(\Pauli_1)$ as a matrix $[H^X|H^Z] \in \mathcal M_{r, 2n}(\F_2)$. 
The rank function can be written $\rank \HH = \rank [H^X|H^Z]$ and the rank of a submatrix is:
$$
\rank \HH_\E = \rank [H^X_\E|H^Z_\E].
$$
From this remark, the proof of the lemma is similar for stabilizer matrices.
\end{proof}

To study the derivatives of $\phi$, we introduce the function $\Phi$ depending on $n$ variables $x = (x_1, x_2, \dots, x_n) \in [0, 1]^n$ defined by:
$$
\Phi(x_1, x_2, \dots, x_n) = \sum_{\E \in \F_2^n} \left[ \left( \rank \HH_\E \right) 
\left( \prod_{\E_i = 1} x_i \right) \left( \prod_{\E_i = 0} (1-x_i) \right) \right].
$$
This function can be seen as the expected rank $\mathbb E_x (\rank \HH_\E )$ for the probability measure such that the $i$-th component of $\E$ is 1 with probability $x_i$ and 0 otherwise, independently of the other components.
This polynomial function $\Phi$ is infinitely derivable and its partial derivatives satisfy:
\begin{lemma} \label{lemma:derivatives}
For all $x$ in $[0, 1]^n$, we have:
$$
\frac{\partial \Phi}{\partial x_i} (x) \geq 0, \qquad \forall i \in \{1, 2, \dots, n\},
$$

$$
\frac{\partial^2 \Phi}{\partial x_i \partial x_j} (x) \leq 0, \qquad \forall i, j \in \{1, 2, \dots, n\}.
$$
\end{lemma}

\begin{proof}
Let $x$ be an element of $[0, 1]^n$.
We remark that $\Phi$ is an affine function in each variable. Thus, we have:
\begin{align*}
\frac{\partial \Phi}{\partial x_i} (x) &= \Phi(x_1, \dots, x_{i-1}, 1, x_{i+1}, \dots, x_n) - \Phi(x_1, \dots, x_{i-1}, 0, x_{i+1}, \dots, x_n)\\
& = \mathbb E_x(\rank (\E \cup \{i\}) - \mathbb E_x(\rank (\E \backslash \{i\})\\
& = \mathbb E_x \Big(\rank (\E \cup \{i\}) - \rank (\E \backslash \{i\} \Big)\\
& \geq 0.
\end{align*}
In the preceding expression, the vector $\E \in \F_2^n$ is considered as a subset of $\{1, 2, \dots, n\}$. Fixing $x_i=1$ is equivalent to replacing the subset $\E$ by $\E \cup \{i\}$ and fixing $x_i=0$ is equivalent to considering the subset $\E \backslash \{i\}$.

Let $j$ be an integer between 1 and $n$ such that $j \neq i$.
The partial derivative $\frac{\partial \Phi}{\partial x_i}(x)$ is also an affine function of the $j$-th variable thus we can derivate it by the same process.
We find:
\begin{align*}
\frac{\partial^2 \Phi}{\partial x_j \partial x_i}(x) &= \frac{\partial}{\partial x_j} \frac{\partial \Phi}{\partial x_i}(x)\\
& = \mathbb E_x( \rank(\E \cup \{i, j\} )) 
- \mathbb E_x( \rank(\E \cup \{i\} \backslash \{j\} ))\\
& - \mathbb E_x( \rank(\E \cup \{j\} \backslash \{i\} ))
+ \mathbb E_x( \rank(\E \backslash \{i, j\} ))\\
& = \mathbb E_x \left[ \rank(A_i \cup A_j) - \rank (A_i) 
- \rank(A_j) + \rank(A_i \cap A_j)  \right]\\
& \leq 0.
\end{align*}
The subset $A_i$ denotes the subset $\E \cup \{i\} \backslash \{j\}$ and $A_j$ denotes the subset $\E \cup \{j\} \backslash \{i\}$.
This quantity is negative by the submodularity of Lemma \ref{lemma:submodularity}

If $j = i$ then the second partial derivative is null.
\end{proof}

\begin{prop} \label{prop:increasing}
Let $\HH \in M_{r, n}(\Pauli_1)$ be a Pauli matrix with $n$ columns.
The function $\phi(p) = \frac 1n \mathbb E_p (\rank \HH_\E)$ is an increasing function on $[0, 1]$.
\end{prop}

\begin{proof}
The function $\phi$ is $\frac 1n \Phi \circ i$ where $i$ is the injection from $[0,1]$ to $[0, 1]^n$ which sends $p$ onto $(p, p \dots, p)$. The derivatives of $\phi$ can be expressed as function of the partial derivatives of $\Phi$:
$$
\phi'(p) = \frac 1n \sum_{i=1}^{n} p \frac{\partial \Phi}{\partial x_i} (i(p)).
$$
From Lemma \ref{lemma:derivatives}, this number is always non-negative.
\end{proof}

\begin{prop} \label{prop:concavity}
Let $\HH \in M_{r, n}(\Pauli_1)$ be a Pauli matrix with $n$ columns.
The function $\phi(p) = \frac 1n \mathbb E_p (\rank \HH_\E)$ is a concave function on $[0, 1]$.
\end{prop}

\begin{proof}
With the same notation as in the proof of Proposition \ref{prop:increasing}, we obtain:
$$
\phi''(p) = \frac 1n \sum_{i, j = 1}^{n} p^2 \frac{\partial^2 \Phi}{\partial x_i \partial x_j} (i(p))
$$
From Lemma \ref{lemma:derivatives}, this number is always non-positive, proving the concavity of $\phi$.
\end{proof}

\begin{prop} \label{prop:bound_g}
Let $\HH \in M_{r, n}(\Pauli_1)$ be a Pauli matrix with $n$ columns.
Assume that $\phi(p) = \frac{1}{n} \mathbb E_p( \rank \HH)$ is upper bounded by $M(p)$.
Then, the function $\Delta(p)=\phi(1-p)-\phi(p)$ admits the lower bound:
$$
\Delta(p) \geq \left( \frac{1-2p}{1-p} \right) \left( \frac{\rank \HH}{n} - M(p) \right)
$$
\end{prop}

\begin{proof}
It suffices to use the concavity of $f$.
Indeed, by concavity the point $(1-p, \phi(1-p))$ is above the segment between $(p, \phi(p))$ and $(1, \phi(1))$. That is:
$$
\phi(1-p) \geq \phi(p) + (1-2p) \left( \frac{\phi(1) - \phi(p)}{1-p} \right).
$$
The inequality follows using the equality $\phi(1) = \rank \HH /n$ and then the upper bound $\phi(p) \leq M(p)$.
\end{proof}

% pour le fichier .bib:
%\bibliographystyle{plain}
%\bibliography{biblio}

\begin{thebibliography}{10}

\bibitem{Al07}
S.A. Aly.
\newblock A class of quantum LDPC codes derived from latin squares and
  combinatorial objects.
\newblock Technical report, Department of Computer Science, Texas A and M
  University, 2007.

\bibitem{Al08}
S.A. Aly.
\newblock A class of quantum LDPC codes constructed from finite geometries.
\newblock In {\em Proc. of Global Telecommunications Conference, 2008. IEEE GLOBECOM
  2008}, pages 1--5, dec. 2008.

\bibitem{BMK09}
S.~K. Baek, P.~Minnhagen, and B.~J. Kim.
\newblock Percolation on hyperbolic lattices.
\newblock {\em Phys. Rev. E} 79:011124, 2009.

\bibitem{Be96}
I.~Benjamini and O.~Schramm.
\newblock Percolation beyond $\Z^d$, many questions and a few answers.
\newblock {\em Electr. Commun. Probab}, 1:71--82, 1996.

\bibitem{BDS97}
C.~H. Bennett, D.~P. DiVincenzo, and J.~A. Smolin.
\newblock Capacities of quantum erasure channels.
\newblock {\em Phys. Rev. Lett.}, 78:3217--3220, 1997.

\bibitem{Be73}
C.~Berge.
\newblock {\em Graphs and Hypergraphs}.
\newblock Elsevier, 1976, 1973.

\bibitem{BM06}
H.~Bombin and M.~A. Martin-Delgado.
\newblock Topological quantum distillation.
\newblock {\em Phys. Rev. Lett.}, 97:180501, 2006.

\bibitem{BM072}
H.~Bombin and M.~A. Martin-Delgado.
\newblock Exact topological quantum order in $d=3$ and beyond: Branyons and
  brane-net condensates.
\newblock {\em Phys. Rev. B}, 75:075103, 2007.

\bibitem{BM073}
H.~Bombin and M.~A. Martin-Delgado.
\newblock Homological error correction: Classical and quantum codes.
\newblock {\em J. Math. Phys.}, 48(5):052105, 2007.

\bibitem{BKLM02}
D.~Burshtein, M.~Krivelevich, S.~Litsyn, and G.~Miller.
\newblock Upper bounds on the rate of LDPC codes.
\newblock {\em IEEE Trans. on Information Theory}, 48:2437--2449, 2002.

\bibitem{CS96}
A.~R. Calderbank and P.~W. Shor,
\newblock Good quantum error-correcting codes exist. 
\newblock {\it Phys. Rev. A}, 54:1098, 1996.

\bibitem{CRSS98}
A.~R. Calderbank, E.~M. Rains, P.~W. Shor, and N.~J.~A. Sloane.
\newblock Quantum error correction via codes over $GF(4)$.
\newblock {\em IEEE Trans. on Information Theory}, 44(4):1369 --1387,
  jul 1998.

\bibitem{COT07}
T.~Camara, H.~Ollivier, and J-P. Tillich.
\newblock A class of quantum LDPC codes: construction and performances under
  iterative decoding.
\newblock In {\em Proc. of IEEE International Symposium on Information Theory, ISIT 2007}, pages 811 --815, june 2007.

\bibitem{CDZ11}
A.~Couvreur, N.~Delfosse, and G.~Z\'emor.
\newblock A construction of quantum LDPC codes from Cayley graphs.
\newblock In  {\em Proc. of IEEE International Symposium on Information Theory, ISIT 2011}, pages 643 --647, 31 2011-aug. 5 2011.

\bibitem{CT91}
T.~M. Cover and J.~A. Thomas.
\newblock {\em Elements of Information Theory}.
\newblock Wiley-Interscience, August 1991.

\bibitem{DZ10}
N.~Delfosse and G.~Z\'emor.
\newblock Quantum erasure-correcting codes and percolation on regular tilings
  of the hyperbolic plane.
\newblock In {\em Proc. of IEEE Information Theory Workshop, ITW 2010}, pages 1--5, sept. 2010.

\bibitem{DKLP02}
E. Dennis, A. Kitaev, A. Landahl, and J. Preskill.
\newblock Topological quantum memory.
\newblock {\em J. Math. Phys.}, 43:4452, 2002.

\bibitem{FS09}
P.~Flajolet and R.~Sedgewick.
\newblock {\em Analytic Combinatorics}.
\newblock Cambridge University Press, 1 edition, 2009.

\bibitem{Ga63}
R.~Gallager.
\newblock {\em Low Density Parity-Check Codes}.
\newblock PhD thesis, Massachusetts Institute of Technology, 1963.

\bibitem{Gr02}
M.~Grassl.
\newblock Algorithmic aspects of quantum error-correcting codes.
\newblock In Ranee~K. Brylinski and Goong Chen, editors, {\em Mathematics of
  Quantum Computation}, pages 223--252. Chapman and Hall/CRC, 2002.

\bibitem{GBP97}
M.~Grassl, T.~Beth, and T.~Pellizzari.
\newblock Codes for the quantum erasure channel.
\newblock {\em Phys. Rev. A}, 56:33--38, Jul 1997.

\bibitem{GZ11}
H.~Gu and R.~M. Ziff.
\newblock Crossing on hyperbolic lattices,  2011.\\
\newblock \url{http://arxiv.org/abs/1111.5626}

\bibitem{HMIH07}
M.~Hagiwara and H.~Imai.
\newblock Quantum quasi-cyclic LDPC codes.
\newblock In {\em Proc. of IEEE International Symposium on Information Theory, ISIT 2007}, pages 806 --810, june 2007.

\bibitem{K80}
H.~Kesten.
\newblock The critical probability of bond percolation on the square lattice
  equals 1/2.
\newblock {\em Communications in Math. Phys.}, 74:41--59, 1980.

\bibitem {Gr} 
  {G. Grimmett},
  {\em Percolation}.\bookhskip
  Springer-Verlag, 1989.

\bibitem{Ki97}
A.~Y. Kitaev.
\newblock Fault-tolerant quantum computation by anyons.
\newblock {\em Ann. Phys.}, 303(1):27, 1997.

\bibitem{MMM04}
D.J.C. MacKay, G.~Mitchison, and P.L. McFadden.
\newblock Sparse-graph codes for quantum error correction.
\newblock {\em IEEE Trans. on Information Theory}, 50(10):2315 -- 2330,
  oct. 2004.

\bibitem{NC00}
M.~A. Nielsen and I.~L. Chuang.
\newblock {\em Quantum Computation and Quantum Information}.
\newblock Cambridge University Press, 1 edition, 2000.

\bibitem{Ox92}
J.~G. Oxley.
\newblock {\em Matroid Theory}.
\newblock Oxford University Press, New York, 1992.

\bibitem{Pr98}
J.~Preskill.
\newblock Quantum information and computation, 1998.\\
\newblock \url{http://www.theory.caltech.edu/people/preskill/ph229/}

\bibitem{RU08}
T.~Richardson and R.~Urbanke.
\newblock {\em Modern Coding Theory}.
\newblock Cambridge University Press, 1 edition, 2008.

\bibitem{SKR08}
K.P. Sarvepalli, M.~Rötteler, and A.~Klappenecker.
\newblock Asymmetric quantum ldpc codes.
\newblock In {\em Proc. of IEEE International Symposium on Information Theory, ISIT 2008}, page 305–309, July 2008.

\bibitem{Si00}
J.~\v{S}ir\'a\v{n}.
\newblock Triangle group representations and constructions of regular maps.
\newblock {\em Proc. of the London Mathematical Society},
  82(03):513--532, 2000.

\bibitem{St96}
A. Steane, 
Multiple particle interference and quantum error correction, {\em Proc. Roy. Soc.
Lond. A}, 452:2551, 1996.

\bibitem{TZ09}
J.-P. Tillich and G.~Z\'emor.
\newblock Quantum LDPC codes with positive rate and minimum distance
  proportional to $n^{1/2}$;.
\newblock In {\em Proc. of IEEE International Symposium on Information Theory, ISIT 2009}, pages 799 --803, july 2009.

\bibitem{Wi06}
H.~S. Wilf.
\newblock {\em Generatingfunctionology}.
\newblock A. K. Peters, Ltd., Natick, MA, USA, 2006.

\bibitem{Ze09}
G. Z\'{e}mor.
\newblock On Cayley graphs, surface codes, and the limits of homological coding
  for quantum error correction.
\newblock In {\em Proc. of the 2nd International Workshop on Coding and
  Cryptology, IWCC '09}, pages 259--273, Berlin, Heidelberg, 2009.
  Springer-Verlag.

\end{thebibliography}

%\newpage
\newcommand{\SortNoop}[1]{}

\end{document}